\documentclass[journal]{IEEEtran}
\IEEEoverridecommandlockouts

\usepackage{silence}
\usepackage{cite}

\usepackage{amsmath,amssymb,amsthm,fixmath}

\usepackage{color,colortbl}
\usepackage{xcolor}

\usepackage[inline]{enumitem}

\usepackage{siunitx}
\DeclareSIUnit{\dBm}{dBm}

\usepackage{graphicx}
\usepackage[labelformat=simple]{subcaption}

\usepackage{epstopdf}
\usepackage{lipsum}

\usepackage{cuted}
\usepackage{multirow,booktabs}
\usepackage{diagbox}
\usepackage{makecell}
\usepackage{hhline}
\usepackage{array} 
\newcolumntype{x}{!{\vrule width 2px}}
\newcolumntype{y}{!{\vrule width 1.5px}}

\usepackage{optidef}

\usepackage[hidelinks]{hyperref}

\usepackage[shortcuts,acronym,automake]{glossaries}
\makeglossaries
\newacronym{6g}{6G}{Sixth Generation}
\newacronym{awgn}{AWGN}{additive white Gaussian noise}
\newacronym{bcd}{BCD}{block coordinate descent}
\newacronym{bec}{BEC}{binary erasure channel}
\newacronym{blcp}{BLCP}{block confusion probability}
\newacronym{blep}{BLEP}{block erasure probability}
\newacronym{fap}{FAP}{false alarm probability}
\newacronym{bler}{BLER}{block error rate}
\newacronym{blec}{BLEC}{block erasure channel}
\newacronym{bri}{BRI}{biregular irreducible}
\newacronym{clt}{CLT}{central limit theorem}
\newacronym{cp}{CP}{control plane}
\newacronym{cdf}{CDF}{cumulative distribution function}
\newacronym{csi}{CSI}{channel state information}
\newacronym{csit}{CSIT}{channel state information at transmitter}
\newacronym{dft-s-ofdm}{DFT-s-OFDM}{Discrete Fourier Transform-spread-OFDM}
\newacronym{fbl}{FBL}{finite blocklength}
\newacronym{gan}{GAN}{generative adversarial network}
\newacronym{ibl}{IBL}{infinite blocklength}
\newacronym{ldpc}{LDPC}{low-density parity-check}
\newacronym{lfp}{LFP}{leakage-failure probability}
\newacronym{ls}{LS}{least squares}
\newacronym{mac}{MAC}{medium access control}
\newacronym{mcs}{MCS}{modulation and coding scheme}
\newacronym{mimo}{MIMO}{multi-input multi-output}
\newacronym{ml}{ML}{maximum likelihood}
\newacronym{mm}{MM}{Minorize-Maximization}
\newacronym{noma}{NOMA}{non-orthogonal multi-access}
\newacronym{nom}{NOM}{non-orthogonal multiplexing}
\newacronym{ofdm}{OFDM}{orthogonal frequency-division multiplexing}
\newacronym{ofdma}{OFDMA}{orthogonal frequency-division multiple access}
\newacronym{oma}{OMA}{orthogonal multiple access}
\newacronym{papr}{PAPR}{Peak-to-Average Power Ratio}
\newacronym{pdf}{PDF}{probability density function}
\newacronym{per}{PER}{packet error rate}
\newacronym{phy}{PHY}{physical}
\newacronym{pld}{PLD}{physical layer deception}
\newacronym{pls}{PLS}{physical layer security}
\newacronym{prb}{PRB}{physical resource block}
\newacronym{psk}{PSK}{phase-shift keying}
\newacronym{sic}{SIC}{successive interference cancellation}
\newacronym{sinr}{SINR}{signal-to-interference-and-noise ratio}
\newacronym{snr}{SNR}{signal-to-noise ratio}
\newacronym{tdma}{TDMA}{time-division multiple access}
\newacronym{up}{UP}{user plane}
\newacronym{urllc}{URLLC}{ultra-reliable low-latency communication}
\newacronym{xurllc}{xURLLC}{next-generation URLLC}

\usepackage{tcolorbox}
\tcbuselibrary{many}
\newtheorem{theorem}{Theorem}% from 'amsthm'
\newtheorem{lemma}{Lemma}% from 'amsthm'
\newtheorem{corollary}{Corollary}% from 'amsthm'
\newtheorem{proposition}{Proposition}% from 'amsthm'
\theoremstyle{definition}
\newtheorem{remark}{Remark}% from 'amsthm'

\usepackage[norelsize,linesnumbered,ruled]{algorithm2e}
\SetKwRepeat{Do}{do}{while}
\makeatletter
\newcommand{\removelatexerror} {\let\@latex@error\@gobble}
\makeatother

\usepackage{flushend}

\newcommand{\superscript}[1]{^{\mathrm{#1}}}
\newcommand{\subscript}[1]{_{\mathrm{#1}}}

\newcommand{\comb}[2]{\begin{pmatrix}#1\\#2\end{pmatrix}}

\usepackage[normalem]{ulem}

\usepackage[textsize=tiny,colorinlistoftodos]{todonotes}
\makeatletter
\define@key{todonotes}{bh}[]{% Comment by 
	\setkeys{todonotes}{author=\textbf{Bin}, color=lime!30}}%
\define@key{todonotes}{yz}[]{% Comment by 
	\setkeys{todonotes}{author=\textbf{Yao}, color=blue!30}}%
\define@key{todonotes}{rs}[]{% Comment by 
	\setkeys{todonotes}{author=\textbf{Rafael}, color=green!30}}%
\makeatother

\tikzstyle{note}=[rectangle, minimum width=3cm, draw = none, fill = none, minimum width = 1.5cm, anchor=center, align=left]
\tikzstyle{block}=[rectangle, draw, line width=1pt, fill = none, minimum width = 1cm, minimum height = 0.75cm, anchor=center, inner sep = 0.5mm, align=center]
\tikzstyle{arrow} = [thick,->,>=stealth]

\newif\ifreviewmode
\reviewmodefalse

\ifreviewmode
\else
  \renewcommand{\todo}[1]{} % hide todo notes
\fi

\begin{document}

\title{Block Erasure Channel and Block z-Channel with Bounded Decoders and Finite Blocklength}

\author{
	\IEEEauthorblockN{
		Bin~Han,~\IEEEmembership{Senior Member,~IEEE,}
        Yao~Zhu,~\IEEEmembership{Member,~IEEE,}
        Rafael~F.~Schaefer,~\IEEEmembership{Senior Member,~IEEE,}
		Wenwen~Chen,
		Giuseppe~Caire,~\IEEEmembership{Fellow,~IEEE,}
        Anke~Schmeink,~\IEEEmembership{Senior Member,~IEEE,}\\        
        H.~Vincent~Poor,~\IEEEmembership{Life Fellow,~IEEE,} and~Hans~D.~Schotten,~\IEEEmembership{Member,~IEEE}
	}%\\
	% \IEEEauthorblockA{
	% 	\IEEEauthorrefmark{1}RPTU University Kaiserslautern-Landau,
    %   \IEEEauthorrefmark{2}Wuhan University,
	%   \IEEEauthorrefmark{2}RWTH Aachen University,
	% 	\IEEEauthorrefmark{3}TU Dresden,
	% 	\IEEEauthorrefmark{4}TU Berlin,\\
	% 	\IEEEauthorrefmark{5}Princeton University,
	% 	\IEEEauthorrefmark{6}German Research Center for Artificial Intelligence (DFKI GmbH)
	% }
	\thanks{B. Han, W. Chen, and H. D. Schotten are with RPTU University Kaiserslautern-Landau. Y. Zhu is with Wuhan University.  A. Schmeink is with RWTH Aachen University. R. F. Schaefer is with TU Dresden. G. Caire is with TU Berlin. H. V. Poor is with Princeton University. H. D. Schotten is with the German Research Center for Artificial Intelligence (DFKI). B. Han (bin.han@rptu.de) and Y. Zhu (yao.zhu@whu.edu.cn) are the corresponding authors. Part of this work was presented at IEEE INFOCOM 2026~\cite{HZS+2026confusions}.
	% \emph{AI disclosure:} In accordance with the IEEE policy on the use of AI, the authors disclose that Anthropic's Claude was employed in a strictly assistive capacity during the preparation of this manuscript---limited to language polishing of author-drafted text, auxiliary verification of the mathematical derivations prepared by the authors, and minor refactoring and debugging assistance of author-designed simulation code. All research ideas, the system model, the analytical framework, the theorems and their proofs, the algorithm design, the experimental methodology, and the interpretation of results were conceived, developed, and verified by the human authors, who take full and sole responsibility for the content and correctness of the paper.
	Anthropic's Claude and OpenAI's GPT assisted with language polishing, auxiliary derivation checking, and minor simulation debugging.
	}% <-this % stops a space
	% \thanks{The work of B. Han, M. A. Habibi, A. Schmeink, and H. D. Schotten was supported in part by the German Federal Ministry of Education and Research in the programme of ``Souver\"an. Digital. Vernetzt.'' joint projects 6G-RIC (16KISK028), Open6GHub (16KISK003K/16KISK004/16KISK012), and Open6GHub+ (16KIS2406). The work of Y. Zhu and Y. Hu was supported in part by the National Key R\&D Program of China under Grant 2023YFE0206600, NSFC Grant 62471341, and by the Hubei Provincial Science and Technology Cooperation Project under Grant 2025EHA040.  The work of V. Sciancalepore was partially supported by SNS JU Project 6G-GOALS (GA no. 101139232).}
}

\maketitle

\begin{abstract}
	Block error rate is a standard metric in finite blocklength (FBL) communication, yet it conflates two qualitatively different failure modes: block confusions, where the decoder selects a wrong codeword, and block erasures, where it declares a loss. Higher-layer protocols treat physical-layer failures as erasures, but this cross-layer assumption lacks FBL justification. We derive a rigorous upper bound and a companion lower-side estimate on the \ac{blcp} and \ac{blep} for bounded-distance decoders over \ac{awgn} channels at finite blocklength, recasting the coding problem as a geometric sphere packing one. We analyze the sensitivity of these bounds to blocklength and signal-to-noise ratio, characterize the envelope of the upper bound, and derive closed-form Chernoff approximations. Extending the model to idle transmission blocks, we bound the \ac{fap} and show that a block z-channel abstraction emerges from the bounded-decoding geometry. Numerical results confirm that confusion and false alarm probabilities lie far below the error rate constraint, providing quantitative physical-layer support for the block erasure channel and block z-channel abstractions assumed in protocol design.
\end{abstract}


\glsresetall

\section{Introduction}\label{sec:intro}
For decades, \ac{bler} has been a standard performance metric for digital communication systems using block coding. It has become particularly important for \ac{urllc} and \ac{xurllc}, which operate in the \ac{fbl} regime, where Shannon capacity loses its relevance and lossless transmission is impractical~\cite{PPV2010channel}.

However, the \ac{bler} does not distinguish between two qualitatively different failure modes. The first, a \emph{block confusion}, occurs when the decoder selects a wrong codeword with high confidence, failing to detect the error. The second, a \emph{block erasure}, occurs when a bounded decoder cannot identify any codeword with sufficient confidence and rejects all candidates, issuing an ``erasure/loss'' as output.

Block erasures and the associated \acp{blec} have been studied in detail in the \ac{ibl} regime, with bounds on erasure probability derived for bounded decoders~\cite{Forney1968exponential,Merhav2008error,SM2010exact}. When the transmitter may be idle in some blocks, a third error mode becomes possible: a \emph{false alarm}, where pure noise is mistakenly decoded as a codeword. If both the confusion probability and the \ac{fap} are negligible, the block channel seen at the decoder output has only three transitions: an idle block is always erased, and an active block is either decoded correctly or erased. This is a z-channel over the codebook extended by the idle block; as the capacity of the z-channel reflects, an outer code operating on blocks can then convey information through the timing of idle blocks as well. Binary erasure channels and binary z-channels have well-understood capacity and error exponents~\cite{DGP+2006capacity,AKT2004extrinsic}, and their block-level counterparts in the \ac{ibl} regime follow from standard asymptotic arguments~\cite{GQ2006coding,Didier2006new}. In the \ac{fbl} regime, however, little effort has been made to distinguish block erasures from block confusions, or to quantify false alarm probabilities,
despite networking protocols modeling \ac{phy} failures as erasures in retransmission schemes, \ac{mac} protocols, and reliability mechanisms. This cross-layer assumption, that block errors at the \ac{phy} translate to detectable losses rather than silent corruptions at upper layers, remains theoretically unvalidated.
In the \ac{ibl} regime this assumption is well founded: Forney's erasure decoding~\cite{Forney1968exponential} and the information-theoretic analysis of hybrid-ARQ protocols in~\cite{CT2001throughput} show that, at sufficiently large blocklength, the undetected-error probability of a code can be made arbitrarily small while its block error rate is dominated by the information outage probability. At finite blocklength no such limit is available, and \ac{fbl} works such as~\cite{HZS+2025semantic} have implicitly assumed the \ac{bler} to be approximately equal to the \ac{blep} without quantifying the two failure modes separately. Evidence supporting this assumption, however, is missing, and the questions of how to quantify the \ac{fbl} \ac{blep} and the \ac{fbl} \ac{fap} remain open.

This work provides a systematic analysis of the \ac{fbl} \ac{blep} and \ac{fap} in \ac{awgn} channels under \ac{bler}-bounded \ac{ml} decoding. Following the classical information-theoretic approach, we map codewords onto a hyper-sphere in the high-dimensional Euclidean signal space, model the decoder's decision regions as hyper-spheres centered at valid codewords, and recast the coding problem as a geometric one under sphere packing constraints. We derive probability bounds for block confusions and block erasures, analyze their sensitivity to system parameters, and show that the resulting channel is a \emph{block erasure channel} where confusion is negligible. We then extend the model to include idle blocks and show that the \ac{fap} is negligible as well, yielding a \emph{block z-channel} whose asymmetry emerges from the bounded-decoding geometry. These results provide quantitative \ac{phy}-layer support for these standard networking abstractions in the \ac{fbl} regime.

The remainder of this paper is organized as follows. Sec.~\ref{sec:related} reviews related work. Sec.~\ref{sec:bounded_decoding} formulates error-bounded block decoding with erasures, Sec.~\ref{sec:bound_analysis} derives an upper bound and a uniform-spectrum estimate on the \ac{blcp}, and Sec.~\ref{sec:sensitivity_analysis} analyzes their sensitivity to blocklength and \ac{snr}. Sec.~\ref{sec:z_channel} extends the model to idle blocks and shows that the resulting block z-channel has negligible \ac{fap}. Sec.~\ref{sec:results} presents numerical and code-based validation, Sec.~\ref{sec:discussion} discusses non-constant-envelope constellations, and Sec.~\ref{sec:conclusion} concludes.

% \clearpage

\section{Related Work}\label{sec:related}
Information theory and coding theory trace back to the work of Shannon~\cite{Shannon1949communication,Shannon1959probability}, who modeled codewords and their decision regions as hyper-spheres packed in Euclidean space, and Hamming~\cite{Hamming1950error}, who introduced the discrete sphere-packing view in the Hamming metric. We follow the former, Euclidean geometric approach throughout.

Classical information theory operates in the asymptotic regime, where typicality arguments achieve capacity but fail to capture the non-asymptotic error behavior of short codes. While early finite-length analyses existed for specific coding schemes~\cite{DPT+2002finite}, systematic \ac{fbl} information theory began with Polyanskiy et al.~\cite{PPV2010channel}, who established tight \ac{bler} bounds for \ac{awgn} channels. These were later extended to multi-antenna fading channels~\cite{YDKP2014quasi} and to both coherent and non-coherent multiple-antenna settings~\cite{AY2019coherent,QK2025noncoherent}; surveys of the \ac{fbl} regime for short-packet \ac{urllc} systems include~\cite{DKP2016toward}. The resulting \ac{bler} bounds have since become a standard tool for evaluating \ac{fbl} codes in cooperative relay networks~\cite{HGS2015capacity}, \ac{noma} schemes~\cite{XYC+2020noma}, and wireless power transfer systems~\cite{AFSA2018wireless}.

Erasure channels have a long history in communications. Forney~\cite{Forney1968exponential} established error exponents for the trade-off between confusions and erasures; these were later refined by Merhav~\cite{Merhav2008error} and Somekh-Baruch and Merhav~\cite{SM2010exact}. These error exponent results characterize the exponential decay rate of the error probability as blocklength grows, operating in a different regime from Polyanskiy-style analysis, which targets achievable rates at fixed blocklength and error probability.
Channels with erasures have since been studied at multiple levels. The \ac{bec} is the simplest nontrivial communication channel: its capacity is elementary, and the analysis of \ac{ldpc} codes over the \ac{bec} takes a particularly simple form~\cite{AKT2004extrinsic,DPT+2002finite}, which was instrumental in the design of the first provably capacity-achieving code families~\cite{LMS+2001efficient,KRU2011threshold}. For $q$-ary erasure channels, \cite{LPC2013bounds} derived performance bounds, and \cite{MTS+2022methods} proposed methods to convert error channels to pure erasure channels. For \acp{blec}, coding performance and bounds are analyzed in~\cite{GQ2006coding,Didier2006new}; a book-length treatment of erasure-channel coding has recently appeared in~\cite{PLM+2026coding}.

A closely related model is the z-channel (also called binary asymmetric channel), in which errors are one-directional: a transmitted~$1$ may be received as~$0$ with probability~$p$, whereas a transmitted~$0$ is always received correctly. Code constructions and capacity bounds for the binary z-channel date back to Kl{\o}ve~\cite{Klove1981asymmetric}; the capacity was characterized by Tallini et~al.~\cite{TAB2002capacity}, who also gave capacity-achieving feedback codes~\cite{TAB2008feedback}. More recently, Polyanskii and Zhang~\cite{PZ2023codes} established new bounds on z-channel code sizes. All of these results, however, treat the binary z-channel at the symbol level. The question of whether a block-level z-channel abstraction, where a block left idle by the transmitter passes through almost undisturbed while an active block faces an erasure probability, is justified in the \ac{fbl} regime has, to the best of our knowledge, not been addressed.

The distinction between confusions and erasures also matters in applications such as semantic communications, where absent information differs qualitatively from distorted information: in~\cite{HZS+2025semantic}, erasures can be exploited while confusions must be avoided, and the opposite preference may arise elsewhere, or even vary across packets within a session.

\section{Error-Bounded Decoding with Erasures}\label{sec:bounded_decoding}
While \emph{complete} decoders always select the most likely codeword, \emph{incomplete} decoders allow an erasure option (selecting none) and/or a list option (listing several candidates)~\cite{Forney1968exponential}. We allow erasures but no list output.

We consider a code $\mathcal{C}\subset\mathcal{A}^n$ of size $\vert\mathcal{C}\vert=M^k$ over an abstract $M$-ary alphabet $\mathcal{A}$ (e.g., $\mathbb{F}_M$ for linear codes), with blocklength $n=k+r$, where $k$ is the payload length and $r$ the redundancy length. A bijective labeling $\mu:\mathcal{A}\to\mathcal{M}$ maps each code symbol onto a constellation $\mathcal{M}$ of $M$ signal points, which occupies $\nu$ real signal dimensions per symbol: $\nu=1$ for real constellations (e.g., BPSK, $\mathcal{M}\subset\mathbb{R}$), and $\nu=2$ for complex ones (e.g., $M$-PSK with $M>2$, $\mathcal{M}\subset\mathbb{C}$, identified with $\mathbb{R}^2$); higher-dimensional constellations such as coherent orthogonal $M$-FSK ($\nu=M$) are covered by the same formalism. Applied symbol by symbol, as in coded modulation with a fixed labeling, $\mu$ maps a codeword $\mathbf{c}=(c_1,\dots,c_n)\in\mathcal{C}$ one-to-one onto the transmitted signal vector $\mathbf{x}=\mu(\mathbf{c})\triangleq(\mu(c_1),\dots,\mu(c_n))$ of the modulated codebook $\mathcal{X}\triangleq\mu(\mathcal{C})\subset\mathcal{M}^n\subset\mathbb{R}^N$, where $N\triangleq\nu n$ is the geometric dimension of the signal space; as $\mu$ is one-to-one, we use the term codeword for both $\mathbf{c}$ and $\mathbf{x}$. The Hamming blocklength $n$ and the geometric dimension $N$ must be carefully distinguished throughout: all combinatorial quantities (Hamming distances, codebook sizes, distance-distribution bounds) are properties of $\mathcal{C}$ and governed by~$n$, whereas all noise statistics (the $\chi$ and non-central $\chi^2$ distributions below) act in $\mathbb{R}^N$ and are governed by~$N$.
Given an ideal channel with \ac{awgn} $\mathbf{w}\sim\mathcal{N}(\mathbf{0},\sigma^2\mathbf{I}_N)$, where $\sigma^2$ denotes the noise variance per real dimension (i.e., the one-sided noise power spectral density is $N_0=2\sigma^2$), the received sequence is $\mathbf{y}=\mathbf{x}+\mathbf{w}$ with the conditional \ac{pdf}\footnote{For complex constellations, Eq.~\eqref{eq:prob} is the \ac{pdf} of the circularly symmetric complex Gaussian noise with variance $2\sigma^2$ per complex symbol, and $\Vert\mathbf{y}-\mathbf{x}\Vert^2=\sum_{i=1}^n\vert y_i-x_i\vert^2$ with $\vert\cdot\vert$ the complex modulus.}
\begin{equation}
	f(\mathbf{y}|\mathbf{x})=\left(2\pi\sigma^2\right)^{-\frac{N}{2}}\exp\left(-\frac{\Vert\mathbf{y}-\mathbf{x}\Vert^2}{2\sigma^2}\right).
	\label{eq:prob}
\end{equation}
We consider a constant-envelope constellation with per-symbol energy $E\subscript{s}=\vert\mu(a)\vert^2$ for all $a\in\mathcal{A}$, so that every codeword has the same energy $\Vert\mathbf{x}\Vert^2=E=nE\subscript{s}$, and equally likely codewords. The constant-envelope assumption is adopted in the main exposition for convenience; Sec.~\ref{sec:discussion} extends the analysis to non-constant-envelope constellations. In Shannon's geometric view~\cite{Shannon1959probability}, all codewords thus lie on the surface of a hyper-sphere of radius $\sqrt{E}$ in $\mathbb{R}^N$. We denote the shortest Euclidean distance between two distinct codewords as $D\subscript{min}$.

Two remarks delimit the class of codes studied. First, the symbol-wise labeling ties the Euclidean distances of Sec.~\ref{sec:bound_analysis} to Hamming distances; all results therefore concern standard coded modulation, in which an $M$-ary code is mapped symbol by symbol onto an $M$-ary constellation. Bit-level constructions such as bit-interleaved and multilevel coded modulation belong to this class at the symbol level, with $m$ coded bits forming one $2^m$-ary label; the bounds then apply to the grouped code. Second, constant-envelope symbol-wise codebooks are spherical codes, but only a small subfamily of the spherical codes of Shannon's analysis~\cite{Shannon1959probability}, whose codewords have components of varying magnitude with an empirical distribution approaching a Gaussian as $N$ grows, unreachable by any symbol-wise labeling of a discrete alphabet. The distinction matters for capacity: general spherical codes achieve the \ac{awgn} capacity at every \ac{snr}, whereas constant-envelope signaling is limited to the polyphase capacity~\cite{Wyner1966bounds}, whose high-\ac{snr} growth is half that of the \ac{awgn} capacity. Capacity is not the subject of this paper, which concerns the split of block errors into confusions and erasures for the coded modulation used in practice; only the false alarm analysis of Sec.~\ref{sec:z_channel}, which uses no Hamming structure, applies to every spherical codebook.

Since $f(\mathbf{y}|\mathbf{x})$ decreases with $\Vert\mathbf{y}-\mathbf{x}\Vert$, \ac{ml} decoding of equally likely codewords is nearest-neighbor decoding. The decoder studied in this paper is a bounded-distance decoder with erasure option: it outputs the codeword $\hat{\mathbf{x}}\in\mathcal{X}$ if $\mathbf{y}$ lies in the ball $\mathcal{B}_R(\hat{\mathbf{x}})\triangleq\{\mathbf{y}\in\mathbb{R}^N:\Vert\mathbf{y}-\hat{\mathbf{x}}\Vert\leqslant R\}$ of radius $R$ around it, and declares an erasure if $\mathbf{y}$ lies in none of these balls. The balls $\mathcal{B}_R(\mathbf{x})$, $\mathbf{x}\in\mathcal{X}$, all of the same radius, are the decision regions of the decoder. Since list output is not allowed, the decision regions must be disjoint, i.e., $R\leqslant D\subscript{min}/2$; we call this the \emph{no-list condition}. Under the no-list condition, a codeword within distance $R$ of $\mathbf{y}$ is necessarily the nearest one, so the decoder coincides with \ac{ml} decoding followed by the acceptance test $\Vert\mathbf{y}-\hat{\mathbf{x}}\Vert\leqslant R$; in this sense it is a bounded \ac{ml} decoder.

Now let an arbitrary codeword $\mathbf{x}=\mu(\mathbf{c})\in\mathcal{X}$ be sent over the noisy channel. Throughout, the Hamming distance $d(\mathbf{c},\mathbf{c}')\triangleq\vert\{i\in\{1,\dots,n\}:c_i\neq c_i'\}\vert$ is taken between code symbol vectors, whereas the Euclidean distance $D(\mathbf{x},\mathbf{x}')\triangleq\Vert\mathbf{x}'-\mathbf{x}\Vert$ is taken between the corresponding signal vectors $\mathbf{x}=\mu(\mathbf{c})$ and $\mathbf{x}'=\mu(\mathbf{c}')$; since $\mu$ is one-to-one, $\mathbf{x}$ and $\mathbf{x}'$ differ exactly in the $d(\mathbf{c},\mathbf{c}')$ symbol positions in which $\mathbf{c}$ and $\mathbf{c}'$ differ (one constellation symbol counting as one position regardless of $\nu$). We denote the set of all neighboring codewords of $\mathbf{x}$ in $\mathcal{X}$ at Hamming distance $\kappa$ as $\mathcal{V}(\mathcal{X},\mathbf{x},\kappa)\triangleq\left\{\mu(\mathbf{c}')\in\mathcal{X}\,\vert\, d(\mathbf{c},\mathbf{c}')=\kappa\right\}$, and therewith the total set of codewords in $\mathcal{X}$ other than $\mathbf{x}$ as $\mathcal{V}_\Sigma(\mathcal{X},\mathbf{x})=\cup_{\kappa>0}\mathcal{V}(\mathcal{X},\mathbf{x},\kappa)$.
When $\mathbf{y}$ falls into the decision region $\mathcal{B}_R(\mathbf{x})$ of the transmitted codeword, the decoding is correct. When it falls into the decision region $\mathcal{B}_R(\mathbf{x}')$ of any other codeword $\mathbf{x}'\in\mathcal{V}_\Sigma(\mathcal{X},\mathbf{x})$, a block confusion occurs. When it falls into no decision region, a block erasure occurs.

Thus, under the \ac{awgn} $\mathbf{w}$, the probability of block errors (including both block confusions and block erasures) is $\varepsilon=P(\Vert\mathbf{w}\Vert\geqslant R)$.

For an $\varepsilon$-bounded decoder, one that accepts the result $\hat{\mathbf{x}}$ only if the estimated block error probability is below $\varepsilon$, the decision region radius is
\begin{equation}\label{eq:decision_region_radius}
	R(\varepsilon)=F^{-1}_{\Vert\mathbf{w}\Vert}(1-\varepsilon)=\sigma F^{-1}_{\chi_N}(1-\varepsilon),
\end{equation}
where $F_{\chi_N}(x)=\mathcal{P}(N/2,x^2/2)$ is the \ac{cdf} of the $\chi$ distribution with $N$ degrees of freedom and $\mathcal{P}(s,x)\triangleq\gamma(s,x)/\Gamma(s)$ is the regularized lower incomplete gamma function, set in calligraphic type to distinguish it from the probability operator $P(\cdot)$.

Meanwhile, the probability of block confusion mistaking $\mathbf{x}$ into $\mathbf{x}'$ is given by $P(\mathbf{x}\to\mathbf{x'})\triangleq P(\hat{\mathbf{x}}=\mathbf{x}'|\mathbf{x})=P(\Vert\mathbf{w}-\mathbf{x'}+\mathbf{x}\Vert\leqslant R)$.
Since $\mathbf{w}\sim\mathcal{N}(\mathbf{0},\sigma^2\mathbf{I}_N)$ is rotationally symmetric, it depends on $\mathbf{x}'-\mathbf{x}$ only through its Euclidean norm, so that
\begin{equation}\label{eq:pairwise_con_prob}
	P(\mathbf{x}\to\mathbf{x}')=P\subscript{pair}(\Vert\mathbf{x}'-\mathbf{x}\Vert).
\end{equation}
To express $P\subscript{pair}(D)$ in closed form, we align the displacement vector with the first coordinate axis of $\mathbb{R}^N$, writing $\mathbf{x}'-\mathbf{x}=D\mathbf{e}_1$ without loss of generality. Then $P\subscript{pair}(D)=P(\Vert\mathbf{w}-D\mathbf{e}_1\Vert^2\leqslant R^2)=P(\Vert\mathbf{w}-D\mathbf{e}_1\Vert^2/\sigma^2\leqslant\rho^2)$, where $\rho\triangleq R/\sigma=F^{-1}_{\chi_N}(1-\varepsilon)$. Since $\mathbf{w}/\sigma\sim\mathcal{N}(\mathbf{0},\mathbf{I}_N)$, the random variable $\Vert\mathbf{w}/\sigma-(D/\sigma)\mathbf{e}_1\Vert^2$ follows a non-central chi-squared distribution with $N$ degrees of freedom and non-centrality parameter $\lambda=D^2/\sigma^2$, so that
\begin{equation}\label{eq:con_prob_upon_distance}
	P\subscript{pair}(D)=F_{\chi^2_N(D^2/\sigma^2)}(\rho^2),
\end{equation}
where $F_{\chi^2_N(\lambda)}(\cdot)$ denotes the \ac{cdf} of the non-central chi-squared distribution with $N$ degrees of freedom and non-centrality parameter $\lambda$.
The overall \ac{blcp} is
\begin{equation}\label{eq:overall_con_prob}
	\begin{split}
		P\subscript{con}=&\sum\limits_{\mathbf{x}\in\mathcal{X}}\left[P(\mathbf{x})\sum\limits_{\mathbf{x}'\in\mathcal{V}_\Sigma(\mathcal{X},\mathbf{x})}P(\mathbf{x}\to\mathbf{x}')\right]
		\\=&\frac{1}{\vert\mathcal{X}\vert}\sum\limits_{\substack{\mathbf{x}\in\mathcal{X}\\\mathbf{x}'\in\mathcal{V}_\Sigma(\mathcal{X},\mathbf{x})}}P(\mathbf{x}\to\mathbf{x}')
	\end{split}
\end{equation}
and the \ac{blep} is $P\subscript{ers}=\varepsilon-P\subscript{con}$, so every lower bound on $P\subscript{con}$ is an upper bound on $P\subscript{ers}$ and vice versa.

For generic \ac{ml} decoders, the dependence of the \ac{bler} and \ac{blep}s on $R$ is straightforward:
\begin{theorem}\label{th:monotonicity_ers_and_con_regarding_R}
	Given fixed $(\mathcal{X}, E,\sigma^2)$, both $\varepsilon$ and $P\subscript{ers}$ are monotonically decreasing w.r.t. $R$, while $P\subscript{con}$ monotonically increases.
\end{theorem}
\begin{proof}
	Trivial, omitted.
\end{proof}

%We focus on the coding aspect: given $(\varepsilon,E,\sigma^2)$ (so that $R$ is fixed via Eq.~\eqref{eq:decision_region_radius}), what is the minimum achievable $P\subscript{ers}$ over all codebooks $\mathcal{X}\subseteq\mathcal{M}^n$? We consider the case $\vert\mathcal{X}\vert=M^k$, i.e., the payload code space is fully utilized.

\section{Confusion Probability Bounds}\label{sec:bound_analysis}
\subsection{Bounds on the Minimum Distance between Codewords}
\label{subsec:bounds_dmin}
Eq.~\eqref{eq:pairwise_con_prob}--\eqref{eq:overall_con_prob} reveal that the \ac{blcp} is dominated by the Euclidean distance between distinct codewords. For two codewords $\mathbf{x}=\mu(\mathbf{c})$ and $\mathbf{x}'=\mu(\mathbf{c}')$ from a constant-modulus constellation with per-symbol energy $E\subscript{s}=|x_i|^2$, the squared Euclidean distance $D^2(\mathbf{x},\mathbf{x}')$ decomposes over the positions where $\mathbf{c}$ and $\mathbf{c}'$ differ:
\begin{equation}\label{eq:euclidean_decomposition}
	D^2(\mathbf{x},\mathbf{x}') = \sum_{i:\,c_i\neq c_i'} |x_i - x_i'|^2 = 2\bigl(1-\bar{\eta}(\mathbf{c},\mathbf{c}')\bigr)E\subscript{s}\,d(\mathbf{c},\mathbf{c}'),
\end{equation}
where $\eta_i \triangleq \text{Re}(x_i x_i'^*)/E\subscript{s}$ is the normalized symbol correlation at position~$i$ and $\bar{\eta}(\mathbf{c},\mathbf{c}')\triangleq d(\mathbf{c},\mathbf{c}')^{-1}\sum_{i:\,c_i\neq c_i'}\eta_i$ is the \emph{empirical mean correlation} of the pair. Eq.~\eqref{eq:euclidean_decomposition} is exact, but $\bar{\eta}(\mathbf{c},\mathbf{c}')$ depends on the symbols that actually appear in the two codewords. To relate $D$ to $d$ through the constellation alone, define the \emph{average pairwise correlation} $\bar{\eta}\triangleq\binom{M}{2}^{-1}\sum_{a\neq b\in\mathcal{M}}\text{Re}(ab^*)/E\subscript{s}$ and the per-symbol squared-distance constant $\Delta\triangleq 2(1-\bar{\eta})E\subscript{s}$; for $M$-PSK, by the roots-of-unity identity $\sum_{l=1}^{M-1}\cos(2\pi l/M)=-1$, $\bar{\eta}=-1/(M-1)$, coinciding with the regular simplex. The ordered pair $(a,b)\in\mathcal{M}^2$ with $a\neq b$ is the \emph{transition type} of a mismatched position carrying symbol $a$ in $\mathbf{x}$ and symbol $b$ in $\mathbf{x}'$; there are $M(M-1)$ such types, and $\bar{\eta}$ is the mean of $\eta$ over them. For a tolerance $\delta\geqslant0$, a codeword pair is called \emph{$\delta$-uniform} if $\vert\bar{\eta}(\mathbf{c},\mathbf{c}')-\bar{\eta}\vert\leqslant\delta$, i.e., if its transition types are mixed closely enough to the constellation average. For every $\delta$-uniform pair, Eq.~\eqref{eq:euclidean_decomposition} yields the two-sided bound
\begin{equation}\label{eq:delta_sandwich}
	2(1-\bar{\eta}-\delta)E\subscript{s}\, d\leqslant D^2(\mathbf{x},\mathbf{x}')\leqslant 2(1-\bar{\eta}+\delta)E\subscript{s}\, d
\end{equation}
with $d=d(\mathbf{c},\mathbf{c}')$, i.e., $D^2=\Delta d$ up to the factor $1\pm\delta/(1-\bar{\eta})$.
\begin{remark}[Uniform mismatch assumption]\label{rem:distance_bounds}
	For simplicity of exposition, we assume $\delta=0$ throughout, so that
	\begin{equation}\label{eq:hamming_to_euclidean}
		D(\mathbf{x},\mathbf{x}') = \sqrt{\Delta\cdot d(\mathbf{c},\mathbf{c}')}
	\end{equation}
	for every pair of codewords. This is an assumption on the code, not an identity. When a nonzero $\delta$ is known, every bound of this paper holds with $\Delta$ replaced by the endpoint of~\eqref{eq:delta_sandwich} that keeps the respective inequality valid, in the same manner as the constellation-level sandwich of Sec.~\ref{sec:discussion}. Because the bounds are applied to each codeword pair with its worst-case endpoint, no average over codes or pairs is ever moved inside a probability; Lemma~\ref{lem:mismatch_concentration} below only quantifies how many pairs of a random code may fall outside the tolerance. Which $\delta$ applies is a property of the code:
	\begin{enumerate}
		\item \emph{Equicorrelated constellations ($\delta=0$ for every code).} If all distinct symbol pairs have the same correlation, as for BPSK ($\bar{\eta}=-1$), the regular simplex ($\bar{\eta}=-1/(M-1)$), and orthogonal signaling ($\bar{\eta}=0$), then $\bar{\eta}(\mathbf{c},\mathbf{c}')=\bar{\eta}$ for every pair and Eq.~\eqref{eq:hamming_to_euclidean} is exact.
		\item \emph{Group codes on $M$-PSK.} The individual $M$-PSK correlations $\cos(2\pi l/M)$ differ from their mean. For a $\mathbb{Z}_M$-linear code with the labeling $\mu(a)=\sqrt{E\subscript{s}}\,e^{\mathrm{j}2\pi a/M}$, $\eta_i=\cos[2\pi(c_i-c_i')/M]$ depends only on the difference symbol, and $\mathbf{c}-\mathbf{c}'$ is itself a codeword; a pair is thus exactly $0$-uniform whenever its difference codeword is \emph{symbol-balanced}, i.e., contains every nonzero element of $\mathbb{Z}_M$ equally often. Pairs whose difference codeword is dominated by a single transition type (e.g., a constant codeword) deviate up to the extremal correlations, so $\delta$ has to be determined from the code; Sec.~\ref{subsec:code_validation} does so for a random $\mathbb{Z}_4$-linear code.
		\item \emph{Random codes.} In the ensemble where each code symbol is drawn independently and uniformly from~$\mathcal{A}$, the transition types at the $d$ mismatched positions are i.i.d.\ uniform, and Lemma~\ref{lem:mismatch_concentration} bounds the probability that a pair violates $\delta$-uniformity at any finite $(d,\delta)$.
		\item \emph{Structured and learned codes.} Interleaving, scrambling, and the polarization transform of practical codes (e.g., polar and turbo codes) spread the transition types without favoring a single one; for these, as for learned (autoencoder-based) codebooks, $\delta$ is best estimated empirically from the codebook.
	\end{enumerate}
\end{remark}

\begin{lemma}\label{lem:mismatch_concentration}
	Let $\mathbf{c},\mathbf{c}'$ be drawn from the i.i.d.-uniform ensemble over $\mathcal{A}^n$, conditioned on $d(\mathbf{c},\mathbf{c}')=d$, and let $\mathbf{x}=\mu(\mathbf{c})$ and $\mathbf{x}'=\mu(\mathbf{c}')$. Then the per-position correlations $\eta_i$ at the mismatched positions are i.i.d.\ with mean $\bar{\eta}$ and range $[\eta\subscript{min},\eta\subscript{max}]$, and
	\begin{equation}\label{eq:hoeffding_uniformity}
		P\!\left(\bigl\vert\bar{\eta}(\mathbf{c},\mathbf{c}')-\bar{\eta}\bigr\vert>\delta\right)\leqslant 2\exp\!\left(-\frac{2d\delta^2}{(\eta\subscript{max}-\eta\subscript{min})^2}\right).
	\end{equation}
\end{lemma}
\begin{proof}
	Conditioned on the mismatch locations, each mismatched position carries an ordered symbol pair drawn uniformly and independently from the $M(M-1)$ transition types, so the $\eta_i$ are i.i.d.\ with mean $\bar{\eta}$ and support $[\eta\subscript{min},\eta\subscript{max}]$, and Eq.~\eqref{eq:hoeffding_uniformity} is Hoeffding's inequality.
\end{proof}
For QPSK ($\eta\subscript{max}-\eta\subscript{min}=1$) at tolerance $\delta=1/3$, the right side of~\eqref{eq:hoeffding_uniformity} equals $0.34$, $0.057$, and $0.0097$ at $d=8$, $16$, and $24$: the concentration argument is effective at the minimum-distance scales of the \ac{fbl} configurations in Sec.~\ref{sec:results} ($d\subscript{min}\superscript{max}$ between $12$ and $192$), while at the shorter distances arising at high \ac{snr} the sandwich~\eqref{eq:delta_sandwich} with a code-specific $\delta$ takes over.

The minimum Hamming distance $d\subscript{min}(\mathcal{C})$ between distinct codewords obeys the following bounds.
\begin{lemma}\label{lem:dmin_bounds}
	Given fixed $(M,n,k,E\subscript{s},\sigma)$, for any $\varepsilon$-bounded decoder that rejects list-output, the minimum Hamming distance $d\subscript{min}(\mathcal{C})$ between any pair of distinct codewords of an $M^k$-sized code $\mathcal{C}$ is always bounded between $[d\subscript{min}\superscript{min}, d\subscript{min}\superscript{max}]$ where
	\begin{align}
		d\subscript{min}\superscript{min}&=\left\lceil\frac{4R^2(\varepsilon)}{\Delta}\right\rceil\label{eq:lower_bound_dmin},\\
		d\subscript{min}\superscript{max}&=\min\left\{d\subscript{min}\superscript{H},\,\left\lfloor\frac{n(M-1)M^{k-1}}{M^k-1}\right\rfloor,\,n-k+1\right\},\label{eq:upper_bound_dmin}
	\end{align}
	with the Hamming-bound term
	\begin{equation}
		d\subscript{min}\superscript{H}\triangleq 2\min\left\{t\in\mathbb{N}\left\vert\sum\limits_{l=0}^t\comb{n}{l}(M-1)^l>M^r\right.\right\},\label{eq:hamming_dmax}
	\end{equation}
	while the second and third terms of Eq.~\eqref{eq:upper_bound_dmin} are the Plotkin and Singleton bounds. The Hamming-bound term typically dominates at moderate code rates, whereas the Plotkin bound is the tightest in the high-distance regime $d\subscript{min}/n>(M-1)/M$ arising at low code rates and low \ac{snr}.
\end{lemma}
\begin{proof}
	See Appendix~\ref{app:dmin_bounds}.
\end{proof}
By Eq.~\eqref{eq:hamming_to_euclidean}, $\sqrt{\Delta\cdot d\subscript{min}\superscript{min}}\leqslant D\subscript{min}\leqslant \sqrt{\Delta\cdot d\subscript{min}\superscript{max}}$.
\subsection{Lower Bound on the Block Confusion Probability}
The monotonicity of $P\subscript{pair}$ w.r.t. $D$ is captured by:
\begin{lemma}\label{lem:monotonicity_Ppair}
	Given fixed $(E\subscript{s},\sigma)$, the pairwise \ac{blcp} $P\subscript{pair}$ is monotonically decreasing in the codeword distance $D=\Vert\mathbf{x}'-\mathbf{x}\Vert$ for all $D\geqslant 2R(\varepsilon)$, and convex if $R(\varepsilon)/\sigma > \sqrt{\ln 3}\,/\,2$.
\end{lemma}
\begin{proof}
	See Appendix~\ref{app:monotonicity_Ppair}.
\end{proof}
This yields a uniform-spectrum estimate of $P\subscript{con}$:
\begin{proposition}
	\label{th:lower_bound_Pcon}
	Given fixed $(M,n,k,E\subscript{s})$, for any $\varepsilon$-bounded decoder that rejects list-output with a well-distributed codebook $\mathcal{X}$ whose average number of minimum-distance neighbors matches the uniform-spectrum estimate $\bar{A}\subscript{min}\approx(M^k/M^n)\comb{n}{d\subscript{min}}(M-1)^{d\subscript{min}}$, the \ac{blcp} $P\subscript{con}$ is approximated from below by
	\begin{equation}\label{eq:lower_bound_Pcon}
		\widetilde{P}\subscript{con}\superscript{LB}\triangleq\comb{n}{d\subscript{min}\superscript{max}}\frac{(M-1)^{d\subscript{min}\superscript{max}}}{M^{n-k}}P\subscript{pair}\left(\sqrt{\Delta\cdot d\subscript{min}\superscript{max}}\right).
	\end{equation}
\end{proposition}
\begin{proof}
	See Appendix~\ref{app:lower_bound_Pcon}.
\end{proof}
The uniform-spectrum assumption holds for random codebooks; for the explicit codes of Sec.~\ref{subsec:code_validation}, which fall short of $d\subscript{min}\superscript{max}$, the estimate lies below the exact value by the corresponding margin. Correspondingly, the \ac{blep} satisfies the approximate upper bound $P\subscript{ers}\lesssim\varepsilon-\widetilde{P}\subscript{con}\superscript{LB}$.

\subsection{Upper Bound on the Block Confusion Probability}
An upper bound on the \ac{blcp} follows similarly:
\begin{theorem}\label{th:upper_bound_Pcon}
	Given fixed $(M,k,E\subscript{s})$, for any $\varepsilon$-bounded decoder rejecting list output, $P\subscript{con}$ is upper-bounded by
	\begin{equation}\label{eq:upper_bound_Pcon}
		P\subscript{con}\superscript{UB}\triangleq (M^k-1)P\subscript{pair}\left(\sqrt{\Delta\cdot d\subscript{min}\superscript{min}}\right).
	\end{equation}
\end{theorem}
\begin{proof}
	See~\cite[App.~D]{HZS+2026confusions}; the chain of inequalities carries over with the per-symbol squared distance $E$ therein replaced by $\Delta$ and the noise dimension by $N$. Equality requires an equidistant (simplex) codebook, which exists only for $\vert\mathcal{X}\vert\leqslant N+1$.
\end{proof}
Again, it lower-bounds the \ac{blep} with $P\subscript{ers}\geqslant \varepsilon-P\subscript{con}\superscript{UB}$.

\section{Sensitivity of the Confusion Probability Estimate and Bound}\label{sec:sensitivity_analysis}
We now study the sensitivity of $\widetilde{P}\subscript{con}\superscript{LB}$ and $P\subscript{con}\superscript{UB}$ to the system parameters; $P\subscript{ers}$ follows from $P\subscript{con}$ under fixed $\varepsilon$.
\subsection{Estimate versus Power}
We first analyze the impact of $E\subscript{s}$ on $\widetilde{P}\subscript{con}\superscript{LB}$:
\begin{theorem}
	\label{th:monotonicity_Pcon_LB_regarding_E}
	Given fixed $(M, k, n, \sigma, \varepsilon)$, the estimate $\widetilde{P}\subscript{con}\superscript{LB}$ is a monotonically decreasing and convex function of $E\subscript{s}$.
\end{theorem}
\begin{proof}
	See~\cite[App.~E]{HZS+2026confusions}, substituting as in the proof of Theorem~\ref{th:upper_bound_Pcon}.
\end{proof}

\subsection{Estimate versus Blocklength}
To study $\widetilde{P}\subscript{con}\superscript{LB}$ as a function of the blocklength $n$, we first need properties of $d\subscript{min}\superscript{max}(n)$ as defined in Eq.~\eqref{eq:upper_bound_dmin}.
\begin{lemma}
	\label{lem:monotonicity_dmin_max_regarding_n}
	Given fixed $(M, k)$, the maximum minimum Hamming distance $d\subscript{min}\superscript{max}$ is a monotonically increasing (discrete) function of $n$.
\end{lemma}
\begin{proof}
	The Hamming-bound term $d\subscript{min}\superscript{H}$ is nondecreasing in $n$ by Pascal's rule~\cite[App.~F]{HZS+2026confusions}, and so are the Plotkin and Singleton terms of Eq.~\eqref{eq:upper_bound_dmin}; hence their minimum $d\subscript{min}\superscript{max}$ is nondecreasing.
\end{proof}

\begin{lemma}
	\label{lem:bound_dmin_max_regarding_n}
	$\forall n\in[k, M^k-1]$, $\exists~\lambda\leqslant\frac{2(M-1)}{M}+\frac{1}{n}$ such that $d\subscript{min}\superscript{max}(n)\leqslant\lambda n$.
	Equality is attained only at $n=M^k-1$, the Hamming-bound-achieving (perfect code) blocklength.
\end{lemma}
\begin{proof}
	The bound was proven for the Hamming-bound term $d\subscript{min}\superscript{H}$ in~\cite[App.~G]{HZS+2026confusions} via the $M$-ary entropy function; since $d\subscript{min}\superscript{max}\leqslant d\subscript{min}\superscript{H}$ by Eq.~\eqref{eq:upper_bound_dmin}, the claim follows.
\end{proof}

\begin{lemma}\label{lem:bound_lambda}
	For all $n\in[k, M^k-2]$, $d\subscript{min}\superscript{max}(n)<\frac{(M-1)(n+1)}{M}$.
\end{lemma}
\begin{proof}
	See Appendix~\ref{app:bound_lambda}.
\end{proof}

\begin{remark}\label{rem:lemma_comparison}
	The Plotkin-based bound of Lemma~\ref{lem:bound_lambda} is strictly tighter throughout its validity range $n\in[k,M^k-2]$: its leading term is $\frac{M-1}{M}n$, versus $\frac{2(M-1)}{M}n$ for the entropy-based bound of Lemma~\ref{lem:bound_dmin_max_regarding_n}, roughly a factor of two. Lemma~\ref{lem:bound_dmin_max_regarding_n} is nevertheless retained because it characterizes the equality case at the perfect-code blocklength $n=M^k-1$, which lies outside the range of Lemma~\ref{lem:bound_lambda}. Lemma~\ref{lem:bound_lambda}, in turn, supplies the strict prefactor inequality $\zeta(n)<1$ used in Theorem~\ref{th:monotonicity_PconLB_regarding_n} below.
\end{remark}

Both $d\subscript{min}\superscript{max}(n)$ and $d\subscript{min}\superscript{min}(n)$ are nondecreasing integer-valued functions of $n$; for such a function, we call each maximal run of consecutive blocklengths on which it is constant a \emph{consistent interval} of that function, and say the function \emph{jumps} at the boundaries between consecutive intervals. The behavior of both the estimate $\widetilde{P}\subscript{con}\superscript{LB}$ and the bound $P\subscript{con}\superscript{UB}$ over $n$ is governed by the interplay of two effects: the $\varepsilon$-calibrated decision radius grows with the blocklength, and the relevant Hamming distance occasionally jumps. The first effect is captured by the following monotonicity property of the pairwise term.

\begin{lemma}\label{lem:Ppair_dimension_monotone}
	Fix $D>0$, $\varepsilon\in(0,1)$, and $\sigma>0$, and let the decision radius be calibrated to $\varepsilon$ at every blocklength per Eq.~\eqref{eq:decision_region_radius}. Then the pairwise confusion probability at fixed Euclidean distance is strictly increasing in the blocklength: $P\subscript{pair}^{(n+1)}(D)>P\subscript{pair}^{(n)}(D)$ for all $n$.
\end{lemma}
\begin{proof}
	See Appendix~\ref{app:Ppair_dimension_monotone}.
\end{proof}

\begin{theorem}
	\label{th:monotonicity_PconLB_regarding_n}
	Given fixed $(M, k, E\subscript{s}, \sigma, \varepsilon)$ and writing $d=d\subscript{min}\superscript{max}(n)$, the estimate evolves over the blocklength as follows.
	\begin{enumerate}
		\item Whenever $d\subscript{min}\superscript{max}(n+1)=d\subscript{min}\superscript{max}(n)$ (within a consistent interval), the step ratio factorizes exactly as
		\begin{equation}\label{eq:PconLB_interval_ratio}
			\frac{\widetilde{P}\subscript{con}\superscript{LB}(n+1)}{\widetilde{P}\subscript{con}\superscript{LB}(n)}
			=\zeta(n)\cdot
			\frac{P\subscript{pair}^{(n+1)}\bigl(\sqrt{\Delta d}\bigr)}{P\subscript{pair}^{(n)}\bigl(\sqrt{\Delta d}\bigr)},
		\end{equation}
		where the combinatorial factor $\zeta(n)\triangleq\frac{n+1}{(n+1-d)M}$ satisfies $\zeta(n)<1$ by Lemma~\ref{lem:bound_lambda} and the geometric factor exceeds one by Lemma~\ref{lem:Ppair_dimension_monotone}.
		\item Whenever $d\subscript{min}\superscript{max}$ jumps from $d$ to $d'>d$ at a boundary, the step ratio equals
		\begin{equation}\label{eq:PconLB_jump_ratio}
			\frac{\comb{n+1}{d'}(M-1)^{d'}}{\comb{n}{d}(M-1)^{d}\,M}\cdot
			\frac{P\subscript{pair}^{(n+1)}\bigl(\sqrt{\Delta d'}\bigr)}{P\subscript{pair}^{(n)}\bigl(\sqrt{\Delta d}\bigr)},
		\end{equation}
		in which the combinatorial prefactor grows at most polynomially in $n$ while the pairwise term collapses hyper-exponentially with the increased distance by Lemma~\ref{lem:monotonicity_Ppair}.
	\end{enumerate}
\end{theorem}
\begin{proof}
	See Appendix~\ref{app:monotonicity_PconLB_regarding_n}.
\end{proof}
In all configurations of Sec.~\ref{sec:results}, the within-interval steps~\eqref{eq:PconLB_interval_ratio} stayed within a factor of four in either direction, while every boundary step~\eqref{eq:PconLB_jump_ratio} dropped the estimate by three to eight orders of magnitude, so that across multiple intervals the estimate decays by hundreds of orders of magnitude, as Fig.~\ref{fig:p2_bounds_regarding_n} shows. The conference precursor asserted a within-interval decrease; the exact factorization~\eqref{eq:PconLB_interval_ratio} corrects this: neither direction holds universally, and the decay is driven entirely by the distance jumps.
% Commented out: LB envelope theorem is less informative than UB envelope since P_con^LB is already monotonically decreasing for large n.
% \begin{theorem}
% 	\label{th:Pcon_LB_envelope_n}
% 	Given fixed $(\varepsilon, M, k, E\subscript{s}, \sigma)$, the value of $\widetilde{P}\subscript{con}\superscript{LB}$ at the boundary $\bar{n}_j$ of the $j$-th $d\subscript{min}\superscript{max}$-consistent interval is
% 	\begin{equation}\label{eq:envelope_LB}
% 		\widetilde{P}\subscript{con}\superscript{LB}(\bar{n}_j) = \Lambda_j \cdot F_{\chi^2_{\bar{n}_j}(\Delta\subscript{max} d_j/\sigma^2)}(\rho_j^2),
% 	\end{equation}
% 	where $d_j = d\subscript{min}\superscript{max}(\bar{n}_j)$, $\Lambda_j = \binom{\bar{n}_j}{d_j}(M-1)^{d_j}/M^{\bar{n}_j-k}$, $\rho_j = F^{-1}_{\chi_{\bar{n}_j}}(1-\varepsilon)$, and $F_{\chi^2_n(\lambda)}(\cdot)$ denotes the \ac{cdf} of the non-central chi-squared distribution. These boundary values trace the piecewise structure of $\widetilde{P}\subscript{con}\superscript{LB}$, with the first-order difference of $\widetilde{P}\subscript{con}\superscript{LB}$ exhibiting a discontinuity at each $\bar{n}_j$.
% \end{theorem}
% \begin{proof}
% 	See Appendix~\ref{app:Pcon_LB_envelope_n}.
% \end{proof}

\subsection{Upper Bound versus Power}
We first consider $d\subscript{min}\superscript{min}$, which governs the upper bound $P\subscript{con}\superscript{UB}$:
\begin{lemma}
	For any fixed $(\sigma, n, \varepsilon)$, $d\subscript{min}\superscript{min}$ is non-increasing in $E\subscript{s}$ and decreases by unit steps at discrete points.
\end{lemma}
\begin{proof}
As defined by Eq.~\eqref{eq:lower_bound_dmin}, $d\subscript{min}\superscript{min}=\lceil 4R^2(\varepsilon)/\Delta\rceil=\bigl\lceil 2\sigma^2[F^{-1}_{\chi_N}(1-\varepsilon)]^2/[(1-\bar{\eta})E\subscript{s}]\bigr\rceil$. Since $F^{-1}_{\chi_N}(1-\varepsilon)$ and $\bar{\eta}$ are constants for fixed $(\sigma, n, \varepsilon)$ and constellation, the ceiling argument is non-increasing in $E\subscript{s}$, and being integer-valued it decreases by unit steps.
\end{proof}

This yields the following behavior of $P\subscript{con}\superscript{UB}$:
\begin{theorem}
	\label{th:Pcon_UB_regarding_E}
	Given fixed $(M, k, n, \sigma, \varepsilon)$, $P\subscript{con}\superscript{UB}$ is piecewise continuous regarding $E\subscript{s}$, with jump discontinuities at $E\subscript{s,i}=2R^2(\varepsilon)/[(1-\bar{\eta})i]$ where $i\in\mathbb{N}^+$; it decreases monotonically in each individual continuous interval and jumps \emph{upward} as $E\subscript{s}$ increases through each $E\subscript{s,i}$ (where $d\subscript{min}\superscript{min}$ steps down from $i+1$ to $i$; cf.\ Corollary~\ref{co:Pcon_UB_regarding_E_extremes}).
\end{theorem}
\begin{proof}
	See~\cite[App.~J]{HZS+2026confusions}, substituting as in the proof of Theorem~\ref{th:upper_bound_Pcon}.
\end{proof}
\begin{corollary}
	\label{co:Pcon_UB_regarding_E_extremes}
	Given fixed $(M, k, n, \sigma, \varepsilon)$, $P\subscript{con}\superscript{UB}$ has its local maxima regarding $E\subscript{s}$ at points $E\subscript{s,i}=2R^2(\varepsilon)/[(1-\bar{\eta})i]$ where $i\in\mathbb{N}^+$, and its local minima at $E\subscript{s,i}^-=\lim\limits_{\delta\to 0}E\subscript{s,i}-\vert\delta\vert$. All local maxima are equal, forming a \emph{constant envelope} (ceiling), while the local minima are monotonically increasing w.r.t. $i$.
\end{corollary}
\begin{proof}
	See~\cite[App.~K]{HZS+2026confusions}, substituting as in the proof of Theorem~\ref{th:upper_bound_Pcon}.
\end{proof}
The common value of all local maxima can be expressed in closed form. Since $\Delta(E\subscript{s,i})\cdot i = 4R^2(\varepsilon)$ is independent of $i$, we have
\begin{equation}\label{eq:envelope_UB_Es}
	P\subscript{con}\superscript{UB}(E\subscript{s,i}) = (M^k-1)\,F_{\chi^2_N(4\rho^2)}(\rho^2),
\end{equation}
where $\rho = F^{-1}_{\chi_N}(1-\varepsilon)$. This is the blocklength envelope of Theorem~\ref{th:Pcon_UB_envelope_n} at fixed $n$: the ceiling is set by the ratio $D/R=2$ whether $E\subscript{s}$ or $n$ varies, and the Chernoff bound~\eqref{eq:envelope_chernoff} applies to it at the fixed~$n$.

\subsection{Upper Bound versus Blocklength}
 On the other hand, since $F^{-1}_{\chi_N}(1-\varepsilon)$ strictly increases with $N=\nu n$, and hence with $n$ (for $\varepsilon<1/2$), for any fixed $(\sigma, E\subscript{s}, \varepsilon)$, $d\subscript{min}\superscript{min}$ is non-decreasing w.r.t. $n$, which ensures with Eq.~\eqref{eq:upper_bound_Pcon} that
\begin{theorem}\label{th:Pcon_UB_regarding_n}
	Given fixed $(M, k, E\subscript{s}, \sigma, \varepsilon)$: within every $d\subscript{min}\superscript{min}$-consistent interval, $P\subscript{con}\superscript{UB}$ strictly increases with $n$; at every boundary between consecutive intervals, where $d\subscript{min}\superscript{min}$ jumps upward by one, $P\subscript{con}\superscript{UB}$ jumps \emph{downward}. The result is a sawtooth pattern whose local maxima sit exactly at the right boundaries of the intervals.
\end{theorem}
\begin{proof}
	Within a $d\subscript{min}\superscript{min}$-consistent interval, $d\subscript{min}\superscript{min}$ is constant while $n$ increases: the argument $D=\sqrt{\Delta\cdot d\subscript{min}\superscript{min}}$ of $P\subscript{pair}$ in Eq.~\eqref{eq:upper_bound_Pcon} is fixed, and by Lemma~\ref{lem:Ppair_dimension_monotone} the pairwise term, and hence $P\subscript{con}\superscript{UB}$, strictly increases with $n$.

	For the boundary between the $i$-th and $(i{+}1)$-th intervals, with right boundaries $n_i$ and $n_{i+1}$, the within-interval increase gives $P\subscript{con}\superscript{UB}(n_i+1)\leqslant P\subscript{con}\superscript{UB}(n_{i+1})$, and Theorem~\ref{th:Pcon_UB_envelope_n} (proven independently) gives $P\subscript{con}\superscript{UB}(n_{i+1})<P\subscript{con}\superscript{UB}(n_i)$. Hence $P\subscript{con}\superscript{UB}(n_i+1)<P\subscript{con}\superscript{UB}(n_i)$: every boundary jump is downward, and the local maxima lie at the right boundaries.
\end{proof}
Intuitively, the downward jumps arise because $P\subscript{pair}$ decreases hyper-exponentially in the Euclidean distance $D$ by Lemma~\ref{lem:monotonicity_Ppair}, so the discrete increase of $D$ at each $d\subscript{min}\superscript{min}$ transition dominates the accumulated within-interval growth of the decision radius.

% \todo[inline]{All appendices can be removed as long as already published in the INFOCOM version}
% \subsection{Upper Bound Envelope of Confusion Probability against Blocklength}\label{subsec:envelope_Pcon_UB}
In contrast to $\widetilde{P}\subscript{con}\superscript{LB}$, whose decay Theorem~\ref{th:monotonicity_PconLB_regarding_n} attributes to the distance jumps, $P\subscript{con}\superscript{UB}$ exhibits a sawtooth pattern: it increases within each $d\subscript{min}\superscript{min}$-consistent interval and then drops sharply when $d\subscript{min}\superscript{min}$ increases. The local maxima occur at the right boundary of each interval, just before $d\subscript{min}\superscript{min}$ jumps upward. The \emph{envelope} of these local maxima admits a compact characterization.

By Eq.~\eqref{eq:lower_bound_dmin}, $d\subscript{min}\superscript{min} = \lceil 4R^2(\varepsilon)/\Delta\rceil$. Since $R(\varepsilon)$ increases with $n$, $d\subscript{min}\superscript{min}$ transitions from $i$ to $i{+}1$ at the blocklength $n_i$ satisfying $4R^2(n_i) = \Delta\cdot i$, i.e., $n_i$ is the right boundary of the $d\subscript{min}\superscript{min}=i$ interval. At each such boundary, the Euclidean distance argument of $P\subscript{con}\superscript{UB}$ in Eq.~\eqref{eq:upper_bound_Pcon} evaluates to
\begin{equation}\label{eq:D_at_jump}
	\sqrt{\Delta\cdot d\subscript{min}\superscript{min}(n_i)} = \sqrt{\Delta\cdot i} = 2R(n_i).
\end{equation}

\begin{theorem}\label{th:Pcon_UB_envelope_n}
	Given fixed $(\varepsilon, M, k)$, the local maxima of $P\subscript{con}\superscript{UB}$ at the right boundaries $\{n_i\}$ of the $d\subscript{min}\superscript{min}$-consistent intervals are given by
	\begin{equation}\label{eq:envelope_exact}
		P\subscript{con}\superscript{UB}(n_i) = (M^k-1)\,F_{\chi^2_{N_i}(4\rho_i^2)}(\rho_i^2),
	\end{equation}
	where $N_i=\nu n_i$ and $\rho_i = F^{-1}_{\chi_{N_i}}(1-\varepsilon)$. This envelope is strictly decreasing, i.e., $P\subscript{con}\superscript{UB}(n_{i+1}) < P\subscript{con}\superscript{UB}(n_i)$ for all $i$, and admits the Chernoff upper bound

	\begin{equation}\label{eq:envelope_chernoff}
		\begin{split}
			P\subscript{con}\superscript{UB}(n_i) < (M^k-1)\,\exp\!\bigg(&-\frac{(u_i^*{-}1)^2}{2}\rho_i^2 \\
			&+ \frac{N_i}{2}\!\left[(u_i^*{-}1) - \ln u_i^*\right]\bigg),
		\end{split}
	\end{equation}
	where $u_i^* = \frac{N_i + \sqrt{N_i^2 + 16\rho_i^4}}{2\rho_i^2}$.
\end{theorem}
\begin{proof}
	See Appendix~\ref{app:Pcon_UB_envelope_n}.
\end{proof}
Eq.~\eqref{eq:envelope_exact} is exact for any finite $n_i$, while the Chernoff bound~\eqref{eq:envelope_chernoff} captures the same exponential decay rate in closed form but omits a polynomial prefactor of order $\Theta(n_i^{-1/2})$, a nearly constant gap on a logarithmic scale. The envelope decay depends only on $(\varepsilon, M, k)$, not on the constellation geometry.

\section{Block z-Channel}\label{sec:z_channel}

The $\varepsilon$-bounded decoder thus turns the physical \ac{awgn} channel into a \emph{\ac{blec}}: each transmitted codeword is either decoded correctly (probability $1-\varepsilon$) or erased (probability $\varepsilon$), with inter-codeword confusion negligible by comparison. We now extend this model to the possibility that the transmitter is \emph{idle} in some blocks. Such activity uncertainty is common in practice: in grant-free and configured-grant uplink access, the receiver cannot know a priori whether a device transmitted in an allocated resource~\cite{SAS+2020grantfree}, \ac{phy} control procedures must detect discontinuous transmission for the same reason, and in event-triggered or semantic transmission schemes~\cite{HZS+2025semantic}, staying silent is itself an informative action; the single-transmitter link considered here is the elementary building block of any such slotted access system. The resulting channel is a \emph{block z-channel}: unlike the classical (binary) z-channel, which is defined axiomatically by $P(0|0)=1$, $P(1|0)=0$, $P(0|1)=\varepsilon$, $P(1|1)=1-\varepsilon$, our block z-channel emerges from the geometry of bounded decoding in \ac{awgn}. Specifically, we show that a false alarm (an idle block mistakenly decoded as a codeword) occurs with probability $P\subscript{fa}\ll\varepsilon$, so the channel exhibits the characteristic z-channel asymmetry: silent blocks pass through almost undisturbed, while active blocks face the erasure probability~$\varepsilon$. The asymmetry is induced by the decoder: the decision regions lie far from the origin, where pure-noise observations concentrate.

\subsection{Model and Error Taxonomy}\label{subsec:z_channel_model}

In an idle block the transmitter sends nothing instead of a codeword $\mathbf{x}\in\mathcal{X}$, which we denote by $\mathbf{x}\subscript{null}=\mathbf{0}$; the received signal is then pure noise, $\mathbf{y}\subscript{null}=\mathbf{w}$.
Since $\mathbf{0}\not\in\mathcal{M}^n$ for any constellation with nonzero symbol energy, the null hypothesis differs from every codeword in all $n$ positions, and for a constant-envelope constellation the per-position squared distance is $\vert x_i\vert^2=E\subscript{s}$ at every position. Hence $D\subscript{null}^2\triangleq D^2(\mathbf{x},\mathbf{x}\subscript{null})=\Vert\mathbf{x}\Vert^2=nE\subscript{s}$ for all $\mathbf{x}\in\mathcal{X}$, a single value for the whole codebook; the same holds for every spherical codebook, constant-envelope or not, since $D\subscript{null}$ is simply the radius of the sphere.\footnote{$D\subscript{null}$ need not exceed $D\subscript{min}$. With $D\subscript{min}^2=\Delta\,d\subscript{min}$, the ratio $D\subscript{null}/D\subscript{min}$ falls below one whenever $d\subscript{min}>n/[2(1-\bar{\eta})]$, i.e., $d\subscript{min}>3n/8$ for QPSK, which is the case whenever the required minimum distance $d\subscript{min}\superscript{min}$ is a large fraction of $n$: six of the eight configurations of Tab.~\ref{tab:false_alarm} (e.g., ${\approx}0.96$ for the $(32,16)$ QPSK configuration at $\gamma=\SI{10}{\dB}$). The false alarm analysis below does not rely on $D\subscript{null}\geqslant D\subscript{min}$; Sec.~\ref{subsec:bound_evaluation} discusses the resulting configurations in which $P\subscript{fa}\superscript{UB}$ exceeds $P\subscript{con}\superscript{UB}$.}

At the receiver, we adopt a \emph{null-unaware} bounded decoder: the $\varepsilon$-bounded decoder of Sec.~\ref{sec:bounded_decoding} is used unchanged. It outputs $\hat{\mathbf{x}}\in\mathcal{X}$ if $\mathbf{y}$ falls in some decision region and an erasure otherwise; an idle transmitter and a decoding failure are thus not distinguished, both resulting in an erasure.

Under this model, three types of block-level errors can occur:
\begin{enumerate}
	\item \textbf{False alarm} ($P\subscript{fa}$): The transmitter is idle, but the receiver decodes some codeword $\hat{\mathbf{x}}\in\mathcal{X}$.
	\item \textbf{Erasure} ($P\subscript{ers}$): The transmitter sends $\mathbf{x}\in\mathcal{X}$, but the receiver declares an erasure.
	\item \textbf{Confusion} ($P\subscript{con}$): The transmitter sends $\mathbf{x}$, but the receiver decodes a different codeword $\hat{\mathbf{x}}\neq\mathbf{x}$.
\end{enumerate}
Types~2 and~3 are identical to the erasure and confusion events analyzed in Secs.~\ref{sec:bound_analysis}--\ref{sec:sensitivity_analysis}. The only new quantity is the \ac{fap}~$P\subscript{fa}$.

\subsection{False Alarm Probability}\label{subsec:false_alarm}

When the transmitter is idle, $\mathbf{y}=\mathbf{w}\sim\mathcal{N}(\mathbf{0},\sigma^2\mathbf{I}_N)$. A false alarm occurs when $\mathbf{w}$ falls into the decision region of any codeword. Since each codeword $\mathbf{x}\in\mathcal{X}$ is at distance $D\subscript{null}=\sqrt{nE\subscript{s}}$ from the origin, the pairwise \ac{fap} for a single codeword is $P\subscript{pair}(D\subscript{null})$ as defined in Eq.~\eqref{eq:con_prob_upon_distance}. By the union bound over all $M^k$ codewords:
\begin{equation}\label{eq:false_alarm_UB}
	P\subscript{fa} \leqslant M^k \cdot P\subscript{pair}\!\left(\sqrt{nE\subscript{s}}\right).
\end{equation}

To show that~\eqref{eq:false_alarm_UB} is negligible compared to~$\varepsilon$, we express $P\subscript{pair}(D\subscript{null})$ in non-central chi-squared form. Normalizing by $\sigma$ and squaring, we have
\begin{equation}\label{eq:Ppair_Dnull_ncx2}
	P\subscript{pair}(D\subscript{null}) = F_{\chi^2_N(n\gamma)}(\rho^2),
\end{equation}
where $\gamma \triangleq E\subscript{s}/\sigma^2$ is the per-symbol \ac{snr} (normalized to the per-dimension noise variance) and $\rho = F^{-1}_{\chi_N}(1-\varepsilon)$ as before. Since $\gamma>0$, the non-central chi-squared mean $N+n\gamma$ exceeds the threshold $\rho^2$ for all practical $(n,\gamma,\varepsilon)$; consequently $P\subscript{pair}(D\subscript{null})$ is bounded by the left tail of a distribution whose mass concentrates well above the threshold. Applying the Chernoff bound for the non-central chi-squared left tail as in the proof of Theorem~\ref{th:Pcon_UB_envelope_n}:
\begin{equation}\label{eq:Ppair_Dnull_chernoff}
	P\subscript{pair}(D\subscript{null}) \leqslant \exp\!\left(-\frac{(u^*{-}1)^2}{2}\rho^2 + \frac{N}{2}\!\left[(u^*{-}1) - \ln u^*\right]\right),
\end{equation}
where
\begin{equation}\label{eq:u_star_fa}
	u^* = \frac{N + \sqrt{N^2 + 4n\gamma\rho^2}}{2\rho^2}
\end{equation}
is the positive root of $\rho^2 u^2 - Nu - n\gamma = 0$. Substituting into~\eqref{eq:false_alarm_UB}, we obtain:
\begin{theorem}[False alarm bound]\label{th:false_alarm_bound}
	For any finite blocklength~$n$, the \ac{fap} of a null-unaware $\varepsilon$-bounded decoder satisfies
	\begin{equation}\label{eq:false_alarm_chernoff}
		P\subscript{fa} \leqslant \exp\!\left(k\ln M - \frac{(u^*{-}1)^2}{2}\rho^2 + \frac{N}{2}\!\left[(u^*{-}1) - \ln u^*\right]\right),
	\end{equation}
	where $u^*$ is given by Eq.~\eqref{eq:u_star_fa} and $\rho = F^{-1}_{\chi_N}(1-\varepsilon)$. Both sides are computable for any given $(n,k,M,\varepsilon,\gamma)$.
\end{theorem}
\begin{proof}
	Combine Eqs.~\eqref{eq:false_alarm_UB} and~\eqref{eq:Ppair_Dnull_chernoff}, take the logarithm.
\end{proof}

To establish that $P\subscript{fa}\ll\varepsilon$ for all blocklengths of interest, we show that the Chernoff bound~\eqref{eq:false_alarm_chernoff} is monotonically decreasing beyond a small threshold blocklength.
\begin{lemma}[Monotonicity of Chernoff bound]\label{lem:Ppair_Dnull_monotone}
	Define $\theta_N\triangleq\rho_N^2/N$ with $\rho_N=F^{-1}_{\chi_N}(1-\varepsilon)$ and $N=\nu n$; $\theta_N$ is strictly decreasing in~$N$ with $\lim_{N\to\infty}\theta_N=1$. For any $n$ with $\theta_N<1+\gamma/\nu$ (equivalently, $\rho_N^2 < N + n\gamma$), the Chernoff bound~\eqref{eq:Ppair_Dnull_chernoff} on~$P\subscript{pair}(D\subscript{null})$ is strictly decreasing in~$n$.
\end{lemma}
\begin{proof}
	Since $\chi^2_N(n\gamma)$ is a sum of $N$~i.i.d.\ $\chi^2_1(\gamma/\nu)$ random variables with mean $1{+}\gamma/\nu$, the Chernoff bound can be written as $P\subscript{pair}(D\subscript{null}) = P(\tfrac{1}{N}\chi^2_N(n\gamma)\leqslant\theta_N) \leqslant e^{-N\Lambda^*(\theta_N)}$, where $\Lambda^*(\theta) = \sup_{t>0}\left[\frac{1}{2}\ln(1{+}2t) + \frac{(\gamma/\nu)\, t}{1+2t} - t\theta\right]$ is the (left-tail) large-deviation probability function. For $\theta < 1{+}\gamma/\nu$, the supremum is attained at $t^*(\theta)>0$ and $\Lambda^*(\theta)>0$, and $\frac{d\Lambda^*}{d\theta} = -t^*(\theta) < 0$ by the envelope theorem, so $\Lambda^*$ is strictly decreasing on~$(-\infty,1{+}\gamma/\nu)$.

	Since $\theta_N$ is strictly decreasing with $\theta_N > 1$ and $\theta_N\to 1 < 1{+}\gamma/\nu$, whenever $\theta_N < 1{+}\gamma/\nu$ both factors in $N\cdot\Lambda^*(\theta_N)$ are positive and increasing: $N$ increases trivially with $n$, and $\Lambda^*(\theta_N)$ increases because $\theta_N$ decreases while $\Lambda^*$ is decreasing. Hence $N\Lambda^*(\theta_N)$ is strictly increasing, and $e^{-N\Lambda^*(\theta_N)}$ is strictly decreasing.
\end{proof}

The condition $\theta_N < 1{+}\gamma/\nu$ holds for all $n\geqslant n^*$, where $n^*$ is the smallest integer with $\rho_{\nu n^*}^2/(\nu n^*) < 1{+}\gamma/\nu$. In practice, $n^*$ is very small (for $\nu=2$ and $\varepsilon=0.05$, $n^*=13$ at $\gamma=1$ and $n^*=4$ at $\gamma=2$). Therefore, once the Chernoff bound~\eqref{eq:false_alarm_chernoff} is evaluated at a single base blocklength $n_0\geqslant n^*$ and found to satisfy $P\subscript{fa}(n_0)\ll\varepsilon$, Lemma~\ref{lem:Ppair_Dnull_monotone} guarantees $P\subscript{fa}(n)\leqslant P\subscript{fa}(n_0)\ll\varepsilon$ for all $n\geqslant n_0$.

The low-\ac{snr} regime in which $P\subscript{fa}$ could become comparable to~$\varepsilon$ is excluded by the no-list feasibility itself:
\begin{corollary}\label{co:feasibility_snr}
	Given $(M,n,k,\varepsilon)$ with $\theta_N=\rho_N^2/N>1$ as in Lemma~\ref{lem:Ppair_Dnull_monotone}, a necessary condition for the existence of any $M^k$-sized codebook admitting an $\varepsilon$-bounded decoder that rejects list output is
	\begin{equation}\label{eq:feasibility_snr}
		\gamma\geqslant\frac{2M\nu\,\theta_N}{(1-\bar{\eta})(M-1)}\cdot\frac{n}{n+1}
		\overset{\text{$M$-PSK}}{=}2\nu\,\theta_N\,\frac{n}{n+1}>2\nu\,\frac{n}{n+1}.
	\end{equation}
\end{corollary}
\begin{proof}
	Feasibility requires $d\subscript{min}\superscript{min}\leqslant d\subscript{min}\superscript{max}$ by Lemma~\ref{lem:dmin_bounds}. Bounding the left side from below by $4R^2(\varepsilon)/\Delta=2\sigma^2\rho_N^2/[(1-\bar{\eta})E\subscript{s}]=2\nu n\,\theta_N/[(1-\bar{\eta})\gamma]$, the right side from above by the Plotkin term of Eq.~\eqref{eq:upper_bound_dmin} relaxed as in Lemma~\ref{lem:bound_lambda}, i.e., by $(M-1)(n+1)/M$, and solving for $\gamma$ yields Eq.~\eqref{eq:feasibility_snr}; for $M$-PSK, $(1-\bar{\eta})(M-1)=M$ since $\bar{\eta}=-1/(M-1)$.
\end{proof}
For complex $M$-PSK constellations ($\nu=2$), the threshold in Eq.~\eqref{eq:feasibility_snr} exceeds $2\nu=4$ (i.e., $E\subscript{s}/N_0=\SI{3}{\dB}$) whenever $\theta_N>(n+1)/n$, i.e., $\rho_N^2>N+2$. This holds for every $n$ at $\varepsilon=0.05$: with $N=2n$, $P(\chi^2_N>N+2)=P(Y\leqslant n-1)$ for $Y\sim\mathrm{Poisson}(n+1)$, which equals $e^{-2}$ and $4e^{-3}$ for $n=1,2$ and exceeds $\frac{1}{2}-\sqrt{2/[\pi(n+1)]}>0.05$ for $n\geqslant3$, since the median of $Y$ is $n+1$~\cite{Choi1994medians} and $P(Y=n+1)=P(Y=n)\leqslant[2\pi(n+1)]^{-1/2}$ by Stirling's formula. Corollary~\ref{co:feasibility_snr} thus excludes all per-symbol \acp{snr} below $\SI{3}{\dB}$ regardless of the code rate, and the rate-dependent Hamming-bound term of Lemma~\ref{lem:dmin_bounds} tightens the threshold further (e.g., to $E\subscript{b}/N_0\approx\SI{6.6}{\dB}$ for the rate-$1$\,bit/symbol setups of Sec.~\ref{sec:results}). Consequently, \acp{snr} low enough for the false alarm bound~\eqref{eq:false_alarm_UB} to approach~$\varepsilon$ lie outside the feasible domain of the decoder under study, so the asymmetry holds wherever the block z-channel model applies.

Tab.~\ref{tab:false_alarm} in Sec.~\ref{sec:results} evaluates the false alarm upper bound for representative \ac{fbl} configurations; in every case it lies many orders of magnitude below~$\varepsilon$. There, Eq.~\eqref{eq:false_alarm_UB} is evaluated exactly via the non-central chi-squared form~\eqref{eq:Ppair_Dnull_ncx2} rather than the Chernoff tightening~\eqref{eq:false_alarm_chernoff}, which is looser by roughly an order of magnitude, consistent with the omitted $\Theta(n^{-1/2})$ prefactor. Together with $P\subscript{ers}+P\subscript{con}=\varepsilon$, this confirms the block z-channel structure with crossover probability~$\varepsilon$ and $P\subscript{fa}\approx 0$.

% \subsection{Litter-Masked Block z-Channel}
%
% Consider a special case where \emph{litter sequence} or \emph{garbage sequence}
%
% \todo[inline]{to be done - further split the codebook into two clusters: the valid codewords and the litter sequences. Upon every codebook, look for the best split}
% NOTE: Litter-masked z-channel removed from this manuscript for space reasons. Reserved for follow-up work.

\section{Numerical Results}\label{sec:results}
We now evaluate the bounds of Secs.~\ref{sec:sensitivity_analysis} and~\ref{sec:z_channel} numerically. Unless otherwise stated, we set $M=4$ (QPSK, hence $\nu=2$ and $N=2n$),~$\varepsilon=0.05$, and $\sigma^2=0.5$ per real dimension (i.e., $N_0=1$), with $\Delta=2(1-\bar\eta)E\subscript{s}=\frac{2M}{M-1}E\subscript{s}=\frac{8}{3}E\subscript{s}$ under the $\delta=0$ assumption of Remark~\ref{rem:distance_bounds}, and tested the bounds under different payload lengths, blocklengths, and \acp{snr}. Throughout, $d\subscript{min}\superscript{max}$ is evaluated as the tightest of the three bounds in Lemma~\ref{lem:dmin_bounds}, and every plotted or tabulated configuration satisfies the no-list feasibility condition $d\subscript{min}\superscript{min}\leqslant d\subscript{min}\superscript{max}$; all sweep ranges begin at the corresponding feasibility onset.

\begin{figure}[!htpb]
	\centering
	\begin{subfigure}{\linewidth}
		\centering
		\includegraphics[width=.82\linewidth, trim=5 5 5 5, clip]{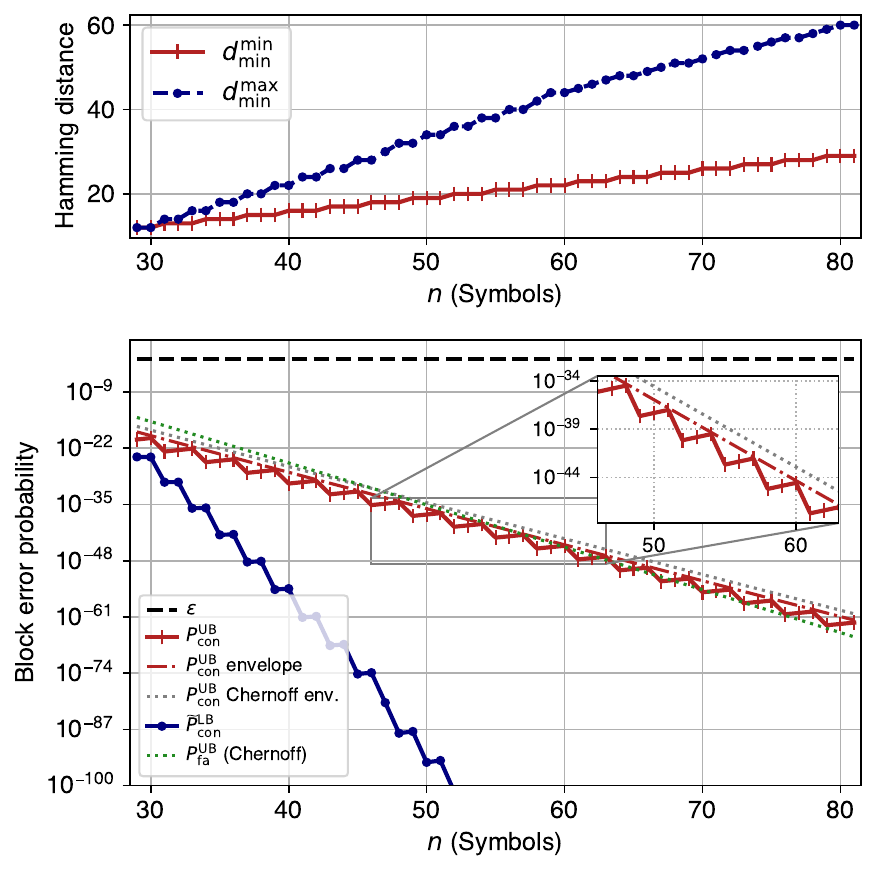}
		\subcaption{$k=16$}
		\label{subfig:p2_bounds_regarding_n_s1}
	\end{subfigure}\\
	\begin{subfigure}{\linewidth}
		\centering
		\includegraphics[width=.82\linewidth, trim=5 5 5 5, clip]{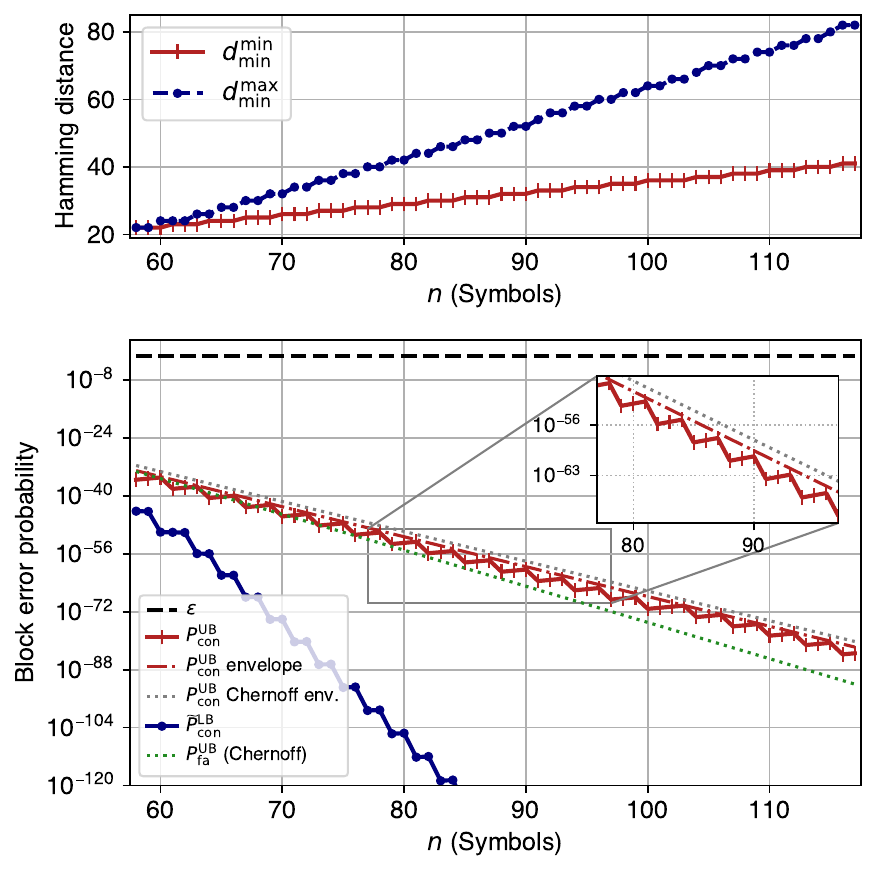}
		\subcaption{$k=32$}
		\label{subfig:p2_bounds_regarding_n_s2}
	\end{subfigure}
	\caption{\Ac{blcp} and \ac{fap} bounds vs.\ $n$ at fixed per-symbol \ac{snr} $E\subscript{s}/N_0=\SI{7}{\dB}$ (i.e., $\gamma=\SI{10}{\dB}$).}
	\label{fig:p2_bounds_regarding_n}
\end{figure}

\begin{figure}[!htpb]
	\centering
	\begin{subfigure}{\linewidth}
		\centering
		\includegraphics[width=.82\linewidth, trim=5 5 5 5, clip]{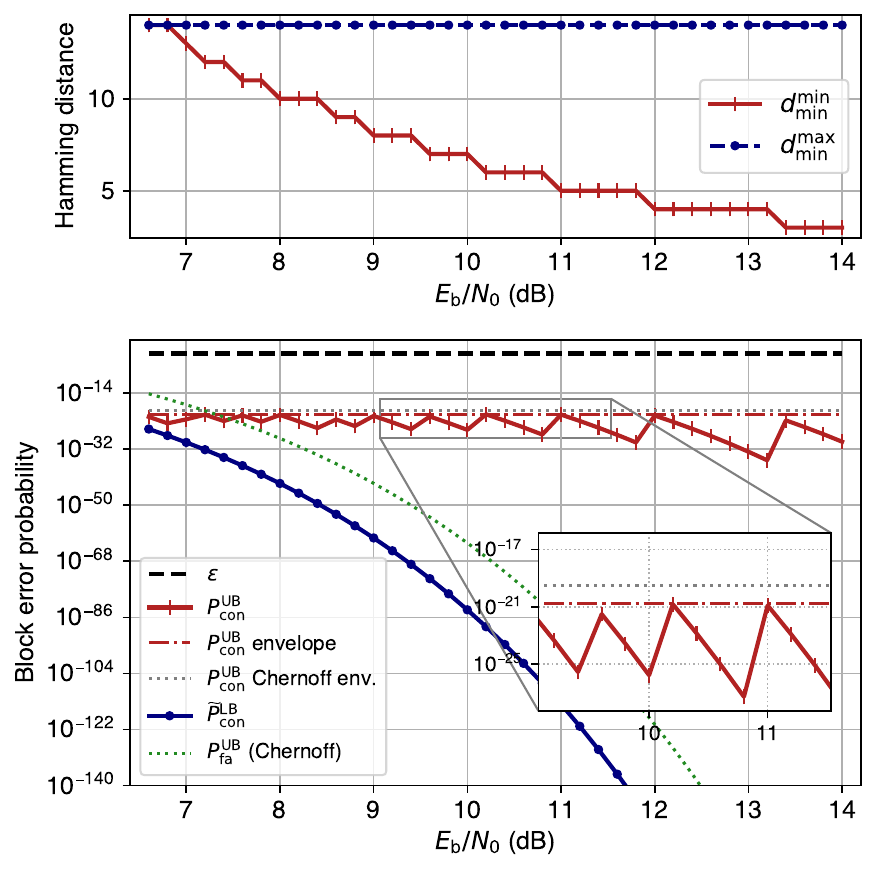}
		\subcaption{With $(32,16)$ codebooks}
		\label{subfig:p2_bounds_regarding_ebn0_s1}
	\end{subfigure}\\
	\begin{subfigure}{\linewidth}
		\centering
		\includegraphics[width=.82\linewidth, trim=5 5 5 5, clip]{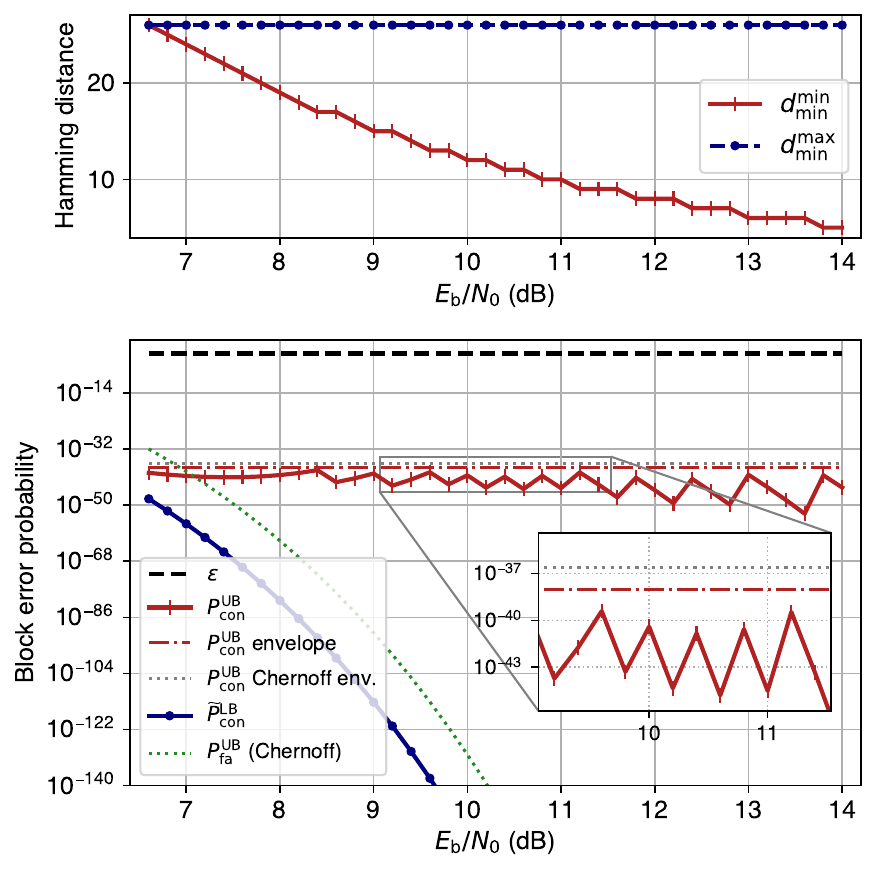}
		\subcaption{With $(64,32)$ codebooks}
		\label{subfig:p2_bounds_regarding_ebn0_s2}
	\end{subfigure}
	\caption{\Ac{blcp} estimate, upper bound, and \ac{fap} upper bound vs.\ the per-information-bit $E\subscript{b}/N_0$.% (both setups carry $1$\,bit per symbol, so $E\subscript{b}/N_0=E\subscript{s}/N_0$); the sweeps begin at the no-list feasibility onset of ${\approx}\SI{6.6}{\dB}$.%
	% Note the $\widetilde{P}\subscript{con}\superscript{LB}$ is convex as proven in Theorem~\ref{th:monotonicity_Pcon_LB_regarding_E}, but may look concave here due to the logarithmic $y$-scale.
	}
	\label{fig:p2_bounds_regarding_ebn0}
\end{figure}

\subsection{Evaluation of the Bounds}\label{subsec:bound_evaluation}
Fig.~\ref{fig:p2_bounds_regarding_n} shows the \ac{blcp} estimate $\widetilde{P}\subscript{con}\superscript{LB}$ and upper bound $P\subscript{con}\superscript{UB}$ versus blocklength $n$ at fixed per-symbol \ac{snr} $\gamma=\SI{10}{\dB}$, for payload lengths $k = 16$ and $k = 32$; the sweeps start at the respective feasibility onsets ($n=29$ and $n=58$). As $n$ increases, both quantities fall by more than a hundred orders of magnitude across the swept ranges, exactly as Theorems~\ref{th:monotonicity_PconLB_regarding_n} and \ref{th:Pcon_UB_regarding_n} describe: the estimate combines mild within-interval variation with dominant multi-order-of-magnitude drops at every $d\subscript{min}\superscript{max}$ transition, while the upper bound follows the proven increase-then-drop sawtooth across the $d\subscript{min}\superscript{min}$-consistent intervals, clearly visible in the zoomed insets.
Both curves remain many orders of magnitude below $\varepsilon = 0.05$, and in particular the rigorous upper bound $P\subscript{con}\superscript{UB}$ does so, confirming that nearly all block errors in this regime are erasures. For system design, this margin confirms that higher-layer protocols can safely model the \ac{phy} as a block erasure channel without additional error detection.
The dash-dotted curve traces the exact envelope of Theorem~\ref{th:Pcon_UB_envelope_n} through the $P\subscript{con}\superscript{UB}$ staircase peaks, while the Chernoff bound~\eqref{eq:envelope_chernoff} (gray dotted) tracks it closely.
Finally, the Chernoff false alarm bound $P\subscript{fa}\superscript{UB}$ of Theorem~\ref{th:false_alarm_bound} (green dotted) decays with~$n$ at a comparable exponential rate; although it can exceed the confusion bounds at the smallest feasible blocklengths, it stays more than $13$ orders of magnitude below~$\varepsilon$ throughout, confirming that false alarms are negligible in this regime.

Fig.~\ref{fig:p2_bounds_regarding_ebn0} plots the confusion probability estimate $\widetilde{P}\subscript{con}\superscript{LB}$ and upper bound $P\subscript{con}\superscript{UB}$ versus the per-information-bit $E_b/N_0$ for codebook setups $(32, 16)$ and $(64, 32)$. The sweeps begin at the no-list feasibility onset ($E\subscript{b}/N_0\approx\SI{6.6}{\dB}$, where $d\subscript{min}\superscript{min}=d\subscript{min}\superscript{max}$, as visible in the upper subplots): below this \ac{snr}, no $M^k$-sized codebook admits an $\varepsilon$-bounded decoder that rejects list output at all. Both curves decrease with increasing \ac{snr}, in agreement with Theorem~\ref{th:monotonicity_ers_and_con_regarding_R}. The step-wise decreases of $d_{\min}^{\min}$ induce upward jump discontinuities in the upper bound, reflecting the discrete nature of distance bounds at fixed decoding radius. As in Fig.~\ref{fig:p2_bounds_regarding_n}, both curves remain orders of magnitude below the \ac{bler} constraint.
Since $n$ and $\varepsilon$ are fixed, the envelope and Chernoff bound form a constant ceiling as given by~\eqref{eq:envelope_UB_Es}; the staircase peaks of $P\subscript{con}\superscript{UB}$ touch this ceiling each time $d\subscript{min}\superscript{min}$ steps down by one.
The Chernoff false alarm bound (green dotted) decreases monotonically with $E\subscript{b}/N_0$ and remains well below~$\varepsilon$ throughout, consistent with the tabulated values in Tab.~\ref{tab:false_alarm}. $P\subscript{fa}\superscript{UB}$ can exceed~$P\subscript{con}\superscript{UB}$ wherever $D\subscript{null}<D\subscript{min}$, i.e., near the feasibility onset of the sweeps and in six of the eight rows of Tab.~\ref{tab:false_alarm}; the z-channel abstraction requires only $P\subscript{fa}\ll\varepsilon$, not $P\subscript{fa}\ll P\subscript{con}$.

Tab.~\ref{tab:false_alarm} evaluates $P\subscript{con}\superscript{UB}$ and $P\subscript{fa}\superscript{UB}$ for feasible \ac{fbl} configurations spanning $M\in\{2,4,8\}$, code rates from $0.12$ to $1$\,bit/symbol, blocklengths from $30$ to $256$, and per-symbol \acp{snr} from $6$ to $\SI{11}{\dB}$. Both bounds fall many orders of magnitude below~$\varepsilon$ in every case; the largest false alarm bound is $3.3\times 10^{-9}$ (the BPSK row). The third row sits at the feasibility edge ($d\subscript{min}\superscript{min}=47$ vs.\ $d\subscript{min}\superscript{max}=48$), and even there $P\subscript{fa}\superscript{UB}<10^{-18}$. Configurations at lower \acp{snr} are infeasible: by Corollary~\ref{co:feasibility_snr}, no $M^k$-sized codebook meets the no-list constraint $2R(\varepsilon)\leqslant D\subscript{min}$ there.

\begin{table}[t]
	\centering
	\caption{Upper bounds $P\subscript{con}\superscript{UB}$ of Theorem~\ref{th:upper_bound_Pcon} and $P\subscript{fa}\superscript{UB}$ of Eq.~\eqref{eq:false_alarm_UB} for feasible \ac{fbl} configurations with $\varepsilon=0.05$.%
	% The tabulated $P\subscript{fa}\superscript{UB}$ values are the exact non-central chi-squared evaluations of Eq.~\eqref{eq:false_alarm_UB} via Eq.~\eqref{eq:Ppair_Dnull_ncx2}, not the Chernoff tightening~\eqref{eq:false_alarm_chernoff}.
	}\label{tab:false_alarm}
	\begin{tabular}{ccccccc}
		\hline
		$M$ & $k$ & $n$ & $k\log_2\!M/n$ & $\gamma$\,[dB] & $P\subscript{con}\superscript{UB}$ & $P\subscript{fa}\superscript{UB}$\\
		\hline
		2 &  4 &  32 & 0.12 &  6 & $<10^{-16}$ & $<10^{-8}$\\
		4 &  6 &  30 & 0.40 &  9 & $<10^{-25}$ & $<10^{-15}$\\
		4 &  6 &  64 & 0.19 &  7 & $<10^{-54}$ & $<10^{-18}$\\
		4 & 16 &  32 & 1.00 & 10 & $<10^{-22}$ & $<10^{-18}$\\
		4 & 32 &  64 & 1.00 & 11 & $<10^{-40}$ & $<10^{-62}$\\
		4 & 32 & 128 & 0.50 &  8 & $<10^{-90}$ & $<10^{-45}$\\
		4 & 32 & 256 & 0.25 &  7 & $<10^{-195}$ & $<10^{-75}$\\
		8 &  6 &  64 & 0.28 & 11 & $<10^{-53}$ & $<10^{-76}$\\
		\hline
	\end{tabular}
\end{table}

\subsection{Validation with Explicit Codebooks}\label{subsec:code_validation}
To connect the bounds to practical codes, we evaluate explicit codebooks under the same $\varepsilon$-bounded \ac{ml} decoder. A structural observation renders this evaluation exact: for any admissible codebook ($2R\leqslant D\subscript{min}$) the spherical decision regions are disjoint, so the sums over neighboring codewords are identities rather than union bounds, i.e., $P\subscript{con}=\sum_{d>0}A_d\,P\subscript{pair}(\sqrt{\Delta d})$ with $\{A_d\}$ the code's distance spectrum, and likewise $P\subscript{fa}=M^k P\subscript{pair}(\sqrt{nE\subscript{s}})$ for constant-envelope codebooks. For BPSK, $\Delta=4E\subscript{s}$ holds exactly at every differing position, so no uniform mismatch assumption is involved. We consider the BCH$(15,7)$ and BCH$(31,16)$ codes ($d\subscript{min}=5$ and $7$, generators constructed from the minimal polynomials over $\mathbb{F}_{2^4}$ and $\mathbb{F}_{2^5}$), the $(32,16)$ polar code with the weight-based (Reed--Muller) information set ($d\subscript{min}=8$), and, in an illustrative role, a PEG-constructed regular-$(3,6)$ LDPC code of the same size ($d\subscript{min}=6$).

\begin{figure}[!htpb]
	\centering
	\includegraphics[width=.85\linewidth]{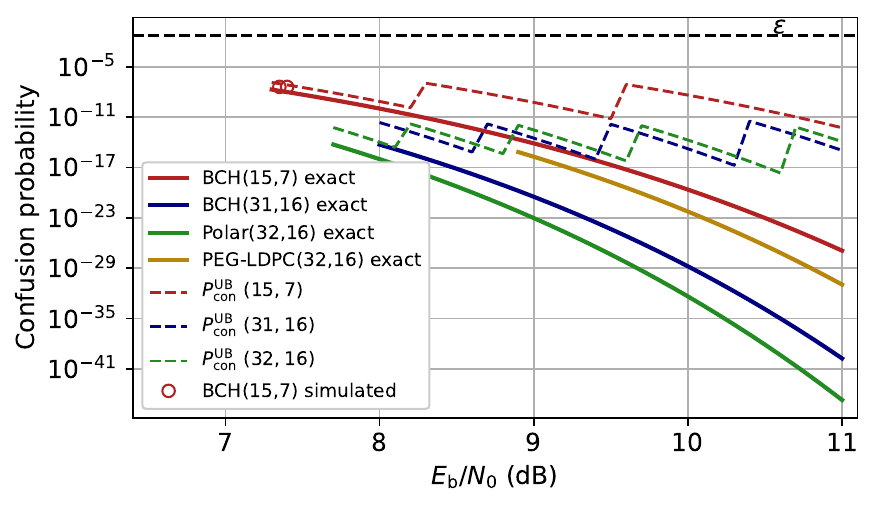}
	\caption{Exact confusion probabilities of explicit codebooks under bounded \ac{ml} decoding, the $(M,n,k)$-only bounds $P\subscript{con}\superscript{UB}$, and simulated values for BCH$(15,7)$ ($5\times10^7$ trials per point); each exact curve starts at the code's admissibility onset.}
	\label{fig:code_validation}
\end{figure}

Fig.~\ref{fig:code_validation} shows the exact confusion probabilities versus the per-information-bit $E\subscript{b}/N_0$, and Tab.~\ref{tab:code_validation} summarizes the reference point $E\subscript{b}/N_0=\SI{8}{\dB}$. Three observations stand out. First, every exact value respects the $(M,n,k)$-only sandwich: it lies below $P\subscript{con}\superscript{UB}$ and above the uniform-spectrum estimate $\widetilde{P}\subscript{con}\superscript{LB}$, the gap to the latter reflecting each code's shortfall from the benchmark $d\subscript{min}\superscript{max}$ (one unit for the BCH codes, two for the polar code). Second, code quality directly governs admissibility: at $\varepsilon=0.05$ the LDPC code's weaker minimum distance postpones its admissibility to $E\subscript{b}/N_0\approx\SI{8.9}{\dB}$ and leaves its confusion probability orders of magnitude above that of the equally-sized polar code, while the same polar blocklength with the Bhattacharyya-designed information set ($d\subscript{min}=4$) is inadmissible below ${\approx}\SI{10.6}{\dB}$; Lemma~\ref{lem:dmin_bounds} thus acts as a design constraint. Third, the Monte-Carlo bounded-\ac{ml} simulation agrees with the exact evaluation within the statistical accuracy of the rare events (two confusions in $5\times10^7$ trials at $\SI{7.35}{\dB}$ against an exact value of $1.4\times10^{-8}$; none at higher \acp{snr}, where the exact values fall below $5\times10^{-9}$), the measured erasure rates satisfy $P\subscript{ers}\approx\varepsilon-P\subscript{con}$ throughout, and decoding pure noise reproduces the exact false alarm values tightly ($5.57\times10^{-5}$ measured vs.\ $5.63\times10^{-5}$ exact at $\SI{7.5}{\dB}$); in every case, all block errors but a vanishing fraction are erasures.

\begin{table}[t]
	\centering
	\caption{Explicit codebooks at $E\subscript{b}/N_0=\SI{8}{\dB}$, $\varepsilon=0.05$, BPSK: minimum distance, smallest admissible $\varepsilon$, and exact confusion and false alarm probabilities.}\label{tab:code_validation}
	\begin{tabular}{lcccc}
		\hline
		Code & $d\subscript{min}$ & $\varepsilon\subscript{min}$ & $P\subscript{con}$ & $P\subscript{fa}$\\
		\hline
		BCH$(15,7)$ & 5 & $0.014$ & $1.0\times10^{-10}$ & $4.0\times10^{-6}$\\
		BCH$(31,16)$ & 7 & $0.044$ & $5.3\times10^{-15}$ & $6.7\times10^{-15}$\\
		Polar-RM$(32,16)$ & 8 & $0.020$ & $9.8\times10^{-17}$ & $1.0\times10^{-14}$\\
		Polar-Bh.$(32,16)$ & 4 & $0.80$ & \multicolumn{2}{c}{inadmissible}\\
		PEG-LDPC$(32,16)$ & 6 & $0.22$ & \multicolumn{2}{c}{inadmissible}\\
		\hline
	\end{tabular}
\end{table}

Two further checks address the modeling assumptions. For the uniform mismatch assumption of Remark~\ref{rem:distance_bounds}, we drew a random $\mathbb{Z}_4$-linear $(16,8)$ codebook on QPSK: $99.5\%$ of all codeword pairs are $\delta$-uniform at $\delta=1/3$ (mean deviation $0.11$), consistent with Lemma~\ref{lem:mismatch_concentration}; the residual outliers are low-weight pairs dominated by a single transition type, for which the sandwich~\eqref{eq:delta_sandwich} with the extremal correlation applies instead. For practical (non-\ac{ml}) decoding, the disjointness of the decision spheres yields a structural guarantee: under the same distance-threshold acceptance test, accepting a wrong codeword requires $\Vert\mathbf{y}-\mathbf{x}'\Vert\leqslant R$, which for an admissible codebook already implies that $\mathbf{x}'$ is the nearest codeword. Hence any candidate selector combined with the acceptance test attains a confusion probability at most equal to the bounded-\ac{ml} value; suboptimal selection only shifts probability mass from correct decisions and confusions into erasures. We confirmed this with bounded \emph{successive-cancellation} decoding of the polar code: at the reference point, $2\times10^5$ trials produced no confusions at an erasure rate of $0.0505$, matching the bounded-\ac{ml} behavior and supporting the \ac{ml}-based bounds as benchmarks for practical decoders.

\section{Non-Constant-Envelope Constellations}\label{sec:discussion}

% \begin{remark}[Generalization beyond constant-envelope constellations]\label{rem:non_constant_envelope}
	The constant-envelope assumption of Sec.~\ref{sec:bounded_decoding} simplifies the parameterization $\Delta = 2(1-\bar{\eta})E\subscript{s}$ but is not required by the bounding framework, which relies only on the symbol-wise labeling.

	The quantities entering the bounds are the extremal per-symbol squared Euclidean distances $\Delta\subscript{min} \triangleq \min_{a\neq b\in\mathcal{M}} |a-b|^2$ and $\Delta\subscript{max} \triangleq \max_{a\neq b\in\mathcal{M}} |a-b|^2$, which yield the Hamming-to-Euclidean sandwich $\Delta\subscript{min}\cdot d(\mathbf{c},\mathbf{c}') \leqslant D^2(\mathbf{x},\mathbf{x}') \leqslant \Delta\subscript{max}\cdot d(\mathbf{c},\mathbf{c}')$ for any pair of distinct codewords $\mathbf{x}=\mu(\mathbf{c})$ and $\mathbf{x}'=\mu(\mathbf{c}')$. This inequality holds for any $M$-ary constellation, including non-constant-envelope ones such as QAM and APSK, as it requires only that every codeword component is drawn from the same alphabet~$\mathcal{M}$.

	The decision radius $R(\varepsilon)=\sigma F^{-1}_{\chi_N}(1-\varepsilon)$ depends only on the noise statistics and~$\varepsilon$, not on the signal energy. The bound and estimate of Proposition~\ref{th:lower_bound_Pcon} and Theorem~\ref{th:upper_bound_Pcon}, the sensitivity analysis of Sec.~\ref{sec:sensitivity_analysis}, and the envelope of Theorem~\ref{th:Pcon_UB_envelope_n} all carry over by substituting $\Delta\subscript{min}$ and $\Delta\subscript{max}$ for the single~$\Delta$. The uniformity ratio $\alpha = \Delta\subscript{min}/\Delta\subscript{max}$ is then the sole constellation-dependent parameter of the envelope. For constant-envelope constellations under the $\delta=0$ mapping of Remark~\ref{rem:distance_bounds}, $\Delta\subscript{min}=\Delta\subscript{max}=\Delta$ and $\alpha=1$.
% \end{remark}

The false alarm analysis of Sec.~\ref{sec:z_channel} extends as well. Theorem~\ref{th:false_alarm_bound} uses only the common distance $D\subscript{null}=\sqrt{nE\subscript{s}}$ of all codewords from the origin and therefore holds for every spherical codebook, whether or not its components have constant modulus. A symbol-wise codebook over a non-constant-envelope constellation is no longer spherical: the codeword energy $\Vert\mathbf{x}\Vert^2=\sum_{i=1}^n\vert x_i\vert^2$ varies across codewords, and the union bound~\eqref{eq:false_alarm_UB} generalizes to
\begin{equation}\label{eq:false_alarm_UB_nce}
	P\subscript{fa}\leqslant\sum_{\mathbf{x}\in\mathcal{X}}P\subscript{pair}\bigl(\Vert\mathbf{x}\Vert\bigr)\leqslant M^k\,P\subscript{pair}\Bigl(\sqrt{n\,E\subscript{s}\superscript{min}}\Bigr),
\end{equation}
where $E\subscript{s}\superscript{min}\triangleq\min_{a\in\mathcal{M}}\vert a\vert^2$ is the minimum symbol energy. The coarse right-hand bound takes the same form as Eq.~\eqref{eq:false_alarm_UB} with $\gamma$ replaced by $\gamma\subscript{min}\triangleq E\subscript{s}\superscript{min}/\sigma^2$; for 16-QAM, $E\subscript{s}\superscript{min}=\bar{E}\subscript{s}/5$ with $\bar{E}\subscript{s}$ the average symbol energy, i.e., a $\SI{7}{\dB}$ penalty relative to the average-energy \ac{snr}. This worst case is attained only by the all-minimum-energy codeword, and the middle term of~\eqref{eq:false_alarm_UB_nce} is tighter. For i.i.d.\ symbols, the fraction of codewords with energy below $n(\bar{E}\subscript{s}-t)$ is at most $\exp[-2nt^2/(E\subscript{s}\superscript{max}-E\subscript{s}\superscript{min})^2]$ by Hoeffding's inequality, but since $P\subscript{pair}$ grows hyper-exponentially as the energy decreases, the sum is dominated by codewords with energy strictly below $n\bar{E}\subscript{s}$, at a deficit set by the balance between their scarcity and their larger false alarm probability. The effective \ac{snr} of the middle term thus lies between $\gamma\subscript{min}$ and $\bar{\gamma}=\bar{E}\subscript{s}/\sigma^2$ and follows exactly from~\eqref{eq:false_alarm_UB_nce} for any codebook; excluding codewords with energy below $nE\subscript{thr}$ at the encoder (energy shaping) raises it to $E\subscript{thr}/\sigma^2$ at a rate loss that vanishes with $n$ for $E\subscript{thr}<\bar{E}\subscript{s}$ by the Hoeffding bound. The false alarm analysis of Sec.~\ref{sec:z_channel} applies with $\gamma$ replaced by this effective \ac{snr}, and the requirement $\mathbf{0}\notin\mathcal{M}$ remains. Tab.~\ref{tab:assumptions} summarizes the modeling assumptions, where each is introduced, and where (or whether) it is relaxed.

\begin{table}[t]
	\centering
	\caption{Modeling assumptions, origins, and scope.}\label{tab:assumptions}
	\begin{tabular}{p{0.32\linewidth}cp{0.42\linewidth}}
		\hline
		Assumption & Origin & Relaxation / scope\\
		\hline
		Constant-envelope constellation & Sec.~\ref{sec:bounded_decoding} & relaxed here via $(\Delta\subscript{min},\Delta\subscript{max})$; $\alpha$ governs the envelope\\
		Symbol-wise one-to-one labeling of an $M$-ary code & Sec.~\ref{sec:bounded_decoding} & needed for the Hamming-to-Euclidean relation; general spherical codes: false alarm bound (Theorem~\ref{th:false_alarm_bound}) only\\
		Uniform symbol mismatch ($\delta=0$) & Rem.~\ref{rem:distance_bounds} & exact for equicorrelated constellations and symbol-balanced difference codewords; otherwise sandwich~\eqref{eq:delta_sandwich} with $\delta$ from Lemma~\ref{lem:mismatch_concentration} or from the code (Sec.~\ref{subsec:code_validation})\\
		Bounded-distance (spherical) decoding, no list output & Sec.~\ref{sec:bounded_decoding} & one-sided benchmark for any selector under the same test (Sec.~\ref{subsec:code_validation}); other acceptance rules not covered\\
		\ac{awgn} isotropy & Sec.~\ref{sec:bounded_decoding} & fading makes the noise geometry direction-dependent\\
		Equiprobable codewords, $\vert\mathcal{X}\vert=M^k$ & Sec.~\ref{sec:bounded_decoding} & non-uniform priors left to future work (Sec.~\ref{sec:conclusion})\\
		$\mathbf{0}\notin\mathcal{M}$ (z-channel only) & Sec.~\ref{sec:z_channel} & holds for PSK, QAM, APSK; verify for non-standard alphabets\\
		\hline
	\end{tabular}
\end{table}

\section{Conclusion and Outlook}\label{sec:conclusion}
We have derived probability bounds for the \ac{blcp} and \ac{blep} of error-bounded decoders at finite blocklength, and analyzed their sensitivity to blocklength and \ac{snr}. Extending the model to idle blocks, we showed that the \ac{fap} is negligible, establishing a block z-channel whose asymmetry emerges from the bounded-decoding geometry. Numerical results across a wide range of \ac{fbl} configurations, together with exact evaluations of explicit BCH, polar, and LDPC codebooks, confirm that the \ac{blcp} lies far below the \ac{bler} constraint and that the \ac{fap} remains well below~$\varepsilon$. These results provide quantitative \ac{phy}-layer support for the block erasure channel and block z-channel abstractions used in higher-layer protocol design.

The bounding framework extends to non-constant-envelope constellations such as QAM and APSK via the extremal per-symbol distances $\Delta\subscript{min}$ and $\Delta\subscript{max}$. Natural directions for future work include non-uniform codeword distributions and exploiting the block erasure and z-channel abstractions in cross-layer protocol design, e.g., informing outer erasure code selection at the transport layer or MAC-layer activity detection that exploits the z-channel asymmetry.

% \section*{Acknowledgment}
% The work of B. Han and H. D. Schotten was in part supported by the Federal Ministry of Research, Technology and Space (BMFTR) of Germany in the project Open6GHub+ (16KIS2406). The work of Y. Zhu and A. Schmeink was in part supported by BMFTR in the project 6GEM+ (16KIS2409K). The work of R. F. Schaefer was in part supported by BMFTR in the project 6G-life (16KIS2413K), and in part supported by the German Research Foundation (DFG) in the Cluster of of Excellence ``Centre for Tactile Internet with Human-in-the-Loop'' (CeTI) of TU Dresden (390696704). The work of G. Caire was in part supported by BMFTR in the project 6G-RIC (16KISK030). The work of H. V. Poor was supported in part by an Innovation Grant from Princeton NextG. B. Han (bin.han@rptu.de) is the corresponding author.

\bibliographystyle{IEEEtran}
\bibliography{references}

\appendices
% \section{Proof of theorem~\ref{th:monotonicity_ers_and_con_regarding_R}}\label{app:monotonicity_ers_and_con_regarding_R}

\section{Proof of Lemma~\ref{lem:dmin_bounds}}\label{app:dmin_bounds}
\begin{proof}
	The lower bound~\eqref{eq:lower_bound_dmin} and the Hamming-bound term~\eqref{eq:hamming_dmax} were established in~\cite[App.~A]{HZS+2026confusions}; the derivations carry over verbatim with the per-symbol squared-distance parameter $E$ therein replaced by $\Delta$ and the noise dimension by $N$. It remains to establish the Plotkin and Singleton terms of Eq.~\eqref{eq:upper_bound_dmin}.

	Two classical arguments provide them. First, the Plotkin bound: in any single symbol position, the number of \emph{ordered} codeword pairs differing there is at most $\vert\mathcal{C}\vert^2(1-1/M)$, since with $m_a$ codewords carrying symbol $a$ in that position, the mismatched ordered pairs number $\vert\mathcal{C}\vert^2-\sum_a m_a^2\leqslant\vert\mathcal{C}\vert^2(1-1/M)$ by the Cauchy--Schwarz inequality. Summing over the $n$ positions and lower-bounding the total by $\vert\mathcal{C}\vert(\vert\mathcal{C}\vert-1)\,d\subscript{min}$ yields, with $\vert\mathcal{C}\vert=M^k$, $d\subscript{min}(\mathcal{C})\leqslant n(M-1)M^{k-1}/(M^k-1)$. Second, the Singleton bound: deleting any $d\subscript{min}-1$ coordinates leaves all $M^k$ codewords distinct, which requires $M^{n-d\subscript{min}+1}\geqslant M^k$, i.e., $d\subscript{min}\leqslant n-k+1$. Taking the minimum of the three bounds (with the Plotkin term rounded down to an integer) establishes Eq.~\eqref{eq:upper_bound_dmin}.
\end{proof}

\section{Proof of Lemma~\ref{lem:monotonicity_Ppair}}\label{app:monotonicity_Ppair}
\begin{proof}
	Without loss of generality, let $\mathbf{d}\triangleq\mathbf{x}'-\mathbf{x}$ with $\Vert\mathbf{d}\Vert=D$, and align $\mathbf{d}=D\mathbf{e}_1$ by exploiting the rotational symmetry of $\mathbf{w}\sim\mathcal{N}(0,\sigma^2 I_N)$. Then $P\subscript{pair}(D) = P((w_1 - D)^2 + \sum_{i=2}^{N} w_i^2 \leqslant R^2)$. Let $S \triangleq \sum_{i=2}^{N} w_i^2/\sigma^2 \sim \chi^2_{N-1}$ (independent of $w_1$), $m \triangleq D/\sigma$, and $\rho \triangleq R/\sigma$. Conditioning on $S$ yields
	\begin{equation}\label{eq:Ppair_decomposed}
		P\subscript{pair}(D) = \mathbb{E}_S\!\left[\mathbf{1}_{S \leqslant \rho^2}\, h\!\left(m, \sqrt{\rho^2 - S}\right)\right],
	\end{equation}
	where $h(m, a) \triangleq \Phi(m + a) - \Phi(m - a)$, with $\Phi$ and $\phi$ denoting the standard normal \ac{cdf} and \ac{pdf}, respectively.

	\emph{Monotonicity.} Differentiating: $\frac{\partial h}{\partial m} = \phi(m + a) - \phi(m - a)$. For $D \geqslant 2R$, we have $m \geqslant 2\rho$ and $a \leqslant \rho$, so $m + a > m - a \geqslant \rho > 0$. Since $\phi$ is strictly decreasing on $(0, \infty)$, $\frac{\partial h}{\partial m} < 0$, thus $\frac{\partial P\subscript{pair}}{\partial D} = \sigma^{-1}\mathbb{E}_S[\mathbf{1}_{S \leqslant \rho^2}\, \frac{\partial h}{\partial m}] < 0$.

	\emph{Convexity.} We compute $\frac{\partial^2 h}{\partial m^2} = -(m + a)\phi(m + a) + (m - a)\phi(m - a)$, which is positive if and only if $f(a) \triangleq (m - a)\, e^{2m a} - (m + a) > 0$. Since $f(0) = 0$ and $f'(a) = e^{2m a}(2m^2 - 2m a - 1) - 1$, we have $f'(0) = 2m^2 - 2 > 0$ for $m > 1$, so $f$ initially increases from zero. The map $a \mapsto e^{2m a}(2m^2 - 2m a - 1)$ has a unique stationary point (a maximum), so $f'$ has at most one zero on $[0, \rho]$: $f$ first increases and may later decrease, hence $f(a) > 0$ on $(0, \rho]$ if and only if $f(\rho) > 0$. Since $a \leqslant \rho \leqslant m/2$, the worst case is $m = 2\rho$ (i.e., $D = 2R$), giving $f(\rho)\vert_{m=2\rho} = \rho(e^{4\rho^2} - 3) > 0 \iff \rho > \sqrt{\ln 3}/2 \approx 0.52$. This condition, $R(\varepsilon)/\sigma = F^{-1}_{\chi_N}(1-\varepsilon) > \sqrt{\ln 3}/2$, is satisfied for all practical \ac{fbl} parameters (any moderate blocklength with $\varepsilon \ll 1$). Under this condition, $\frac{\partial^2 h}{\partial m^2} \geqslant 0$ for all $a \in [0, \rho]$ (with equality only at $a = 0$), and therefore $\frac{\partial^2 P\subscript{pair}}{\partial D^2} = \sigma^{-2}\mathbb{E}_S[\mathbf{1}_{S \leqslant \rho^2}\, \frac{\partial^2 h}{\partial m^2}] > 0$.
\end{proof}

\section{Proof of Lemma~\ref{lem:Ppair_dimension_monotone}}\label{app:Ppair_dimension_monotone}
\begin{proof}
	Work with $\sigma$-normalized quantities and write $N=\nu n$. Let $U\sim\chi^2_N$ and $Y\sim\chi^2_\nu$ be independent, let $T\sim\chi^2_N(\lambda)$ with $\lambda=D^2/\sigma^2>0$ be independent of $Y$, and abbreviate $a=\rho_N^2$ and $b=\rho_{N+\nu}^2$ for the squared normalized decision radii at blocklengths $n$ and $n+1$, respectively. By Eq.~\eqref{eq:con_prob_upon_distance} and the additivity of independent chi-squared variables, $P\subscript{pair}^{(n)}(D)=F_T(a)$ and $P\subscript{pair}^{(n+1)}(D)=\mathbb{E}_Y[F_T(b-Y)]$, while the $\varepsilon$-calibration of the decision radius at both blocklengths gives
	\begin{equation}\label{eq:calibration_identity}
		\mathbb{E}_Y\!\left[F_U(b-Y)\right]=P(U+Y\leqslant b)=1-\varepsilon=F_U(a).
	\end{equation}
	The non-central chi-squared family has a monotone likelihood ratio in its argument: expanding the non-central density as a Poisson mixture, $f_T(t)/f_U(t)=e^{-\lambda/2}\sum_{j\geqslant 0}[(\lambda/2)^j/j!]\,f_{\chi^2_{N+2j}}(t)/f_{\chi^2_N}(t)$, and each ratio $f_{\chi^2_{N+2j}}(t)/f_{\chi^2_N}(t)\propto t^{j}$ is nondecreasing in $t$ (strictly increasing for $j\geqslant1$), so $f_T/f_U$ is strictly increasing for $\lambda>0$. Consequently $\varphi\triangleq F_T\circ F_U^{-1}$ is strictly convex on $(0,1)$, since $\varphi'(u)=(f_T/f_U)\bigl(F_U^{-1}(u)\bigr)$ is strictly increasing in $u$. Noting that $F_T(x)=\varphi\bigl(F_U(x)\bigr)$ for all $x\geqslant0$ and applying Jensen's inequality to the non-degenerate random variable $F_U(b-Y)$, we obtain $P\subscript{pair}^{(n+1)}(D)=\mathbb{E}[\varphi(F_U(b-Y))]>\varphi(\mathbb{E}[F_U(b-Y)])=\varphi(F_U(a))=F_T(a)=P\subscript{pair}^{(n)}(D)$, where the inequality is strict by the strict convexity of $\varphi$.
\end{proof}

\section{Proof of Proposition~\ref{th:lower_bound_Pcon}}\label{app:lower_bound_Pcon}
\begin{proof}
	Restricting the sum in Eq.~\eqref{eq:overall_con_prob} to pairs at the minimum Hamming distance $d\subscript{min}$ of $\mathcal{C}$ drops only non-negative terms, so $P\subscript{con}\geqslant\vert\mathcal{X}\vert^{-1}\sum_{\mathbf{x}\in\mathcal{X},\,\mathbf{x}'\in\mathcal{V}(\mathcal{X},\mathbf{x},d\subscript{min})}P\subscript{pair}(D(\mathbf{x},\mathbf{x}'))$.
	For each pair at Hamming distance $d\subscript{min}$, by Eq.~\eqref{eq:hamming_to_euclidean}, $D(\mathbf{x},\mathbf{x}')=\sqrt{\Delta\cdot d\subscript{min}}$. Since $d\subscript{min}\leqslant d\subscript{min}\superscript{max}$ by Lemma~\ref{lem:dmin_bounds} and $P\subscript{pair}$ is monotonically decreasing by Lemma~\ref{lem:monotonicity_Ppair}, each term satisfies $P\subscript{pair}(D(\mathbf{x},\mathbf{x}'))\geqslant P\subscript{pair}(\sqrt{\Delta\cdot d\subscript{min}\superscript{max}})$, so that
	\begin{equation}\label{eq:lower_bound_Pcon_step2}
		P\subscript{con}\geqslant P\subscript{pair}\!\left(\sqrt{\Delta\cdot d\subscript{min}\superscript{max}}\right)\cdot\bar{A}\subscript{min},
	\end{equation}
	where $\bar{A}\subscript{min}\triangleq\vert\mathcal{X}\vert^{-1}\sum_{\mathbf{x}\in\mathcal{X}}\vert\mathcal{V}(\mathcal{X},\mathbf{x},d\subscript{min})\vert$ is the average number of minimum-distance neighbors per codeword. So far, the chain of inequalities is rigorous.

	It remains to estimate $\bar{A}\subscript{min}$. For a well-distributed code whose $M^k$ codewords are approximately uniformly spread in $\mathcal{A}^n$, the number of codewords at Hamming distance $\kappa$ from any reference codeword is approximately $M^k/M^n$ times the total number of $M$-ary sequences at that distance. Evaluating this uniform-spectrum estimate at $\kappa=d\subscript{min}\superscript{max}$:
	\begin{equation}\label{eq:neighbor_count_estimate}
		\bar{A}\subscript{min}\approx\frac{1}{M^{n-k}}\comb{n}{d\subscript{min}\superscript{max}}(M-1)^{d\subscript{min}\superscript{max}}.
	\end{equation}
	Substituting this estimate into Eq.~\eqref{eq:lower_bound_Pcon_step2} yields Eq.~\eqref{eq:lower_bound_Pcon}. Eq.~\eqref{eq:neighbor_count_estimate} is an approximation, not a one-sided bound: $\widetilde{P}\subscript{con}\superscript{LB}$ is an \emph{estimate} for codebooks that attain $d\subscript{min}\superscript{max}$, and Sec.~\ref{subsec:code_validation} quantifies its gap for explicit codes.
\end{proof}

\section{Proof of Lemma~\ref{lem:bound_lambda}}\label{app:bound_lambda}
\begin{proof}
	By the Plotkin term of Eq.~\eqref{eq:upper_bound_dmin},
	% \begin{equation}
		$d\subscript{min}\superscript{max}(n)\leqslant\frac{n(M-1)M^{k-1}}{M^k-1}$.
	% \end{equation}
	Since $n\leqslant M^k-2$, we have $\frac{n}{M^k-1}<1$, thus
	% \begin{equation}
	% 	\frac{n(M-1)M^{k-1}}{M^k-1}=\frac{M-1}{M}\cdot\frac{nM^k}{M^k-1}<\frac{(M-1)(n+1)}{M}.
	% \end{equation}
	$\frac{n(M-1)M^{k-1}}{M^k-1}=\frac{M-1}{M}\cdot\frac{nM^k}{M^k-1}<\frac{(M-1)(n+1)}{M}$.
\end{proof}

\section{Proof of Theorem~\ref{th:monotonicity_PconLB_regarding_n}}\label{app:monotonicity_PconLB_regarding_n}
\begin{proof}
	Write $P\subscript{pair}^{(n)}(D)$ for the pairwise confusion probability at blocklength~$n$. Within a $d\subscript{min}\superscript{max}$-consistent interval, i.e., when $d\subscript{min}\superscript{max}(n+1)=d\subscript{min}\superscript{max}(n)\triangleq d$, substituting Eq.~\eqref{eq:lower_bound_Pcon} at blocklengths $n+1$ and $n$ and using $\binom{n+1}{d}/\binom{n}{d}=(n+1)/(n+1-d)$ yields Eq.~\eqref{eq:PconLB_interval_ratio} directly. The prefactor bound $\zeta(n)<1$ is equivalent to $d<(n+1)(M-1)/M$, which is exactly Lemma~\ref{lem:bound_lambda}, and the geometric factor exceeds one by Lemma~\ref{lem:Ppair_dimension_monotone}.

	At a boundary where $d\subscript{min}\superscript{max}(n+1)=d'>d=d\subscript{min}\superscript{max}(n)$, the same substitution yields Eq.~\eqref{eq:PconLB_jump_ratio}. Its combinatorial prefactor is bounded by $\binom{n+1}{d'}(M-1)^{d'}/[\binom{n}{d}(M-1)^{d}M]\leqslant [(n+1)(M-1)]^{d'-d}\zeta(n)$, polynomial in $n$ for bounded jump sizes $d'-d$ (one or two in Sec.~\ref{sec:results}), while the pairwise factor contracts hyper-exponentially in $\sqrt{\Delta d'}-\sqrt{\Delta d}$ by Lemma~\ref{lem:monotonicity_Ppair}.
\end{proof}

% Commented out: LB envelope proof (see corresponding theorem comment in main text)
% \section{Proof of Theorem~\ref{th:Pcon_LB_envelope_n}}\label{app:Pcon_LB_envelope_n}
% \begin{proof}
% 	At the last blocklength $\bar{n}_j = n_j - 1$ before the $j$-th jump of $d\subscript{min}\superscript{max}$, we have $d\subscript{min}\superscript{max}(\bar{n}_j) = d_j$ and, from Eq.~\eqref{eq:lower_bound_Pcon},
% 	\begin{equation}
% 		\widetilde{P}\subscript{con}\superscript{LB}(\bar{n}_j) = \Lambda_j \cdot P\subscript{pair}^{(\bar{n}_j)}\!\left(\sqrt{\Delta\subscript{max}\cdot d_j}\right).
% 	\end{equation}
% 	Recall from Eq.~\eqref{eq:con_prob_upon_distance} that $P\subscript{pair}(D) = P(\Vert\mathbf{w}-\mathbf{d}\Vert \leqslant R)$ where $\mathbf{w}\sim\mathcal{N}(\mathbf{0},\sigma^2\mathbf{I}_{\bar{n}_j})$ and $\Vert\mathbf{d}\Vert = D$. Normalizing by $\sigma$ and squaring, this equals $P(X \leqslant \rho_j^2)$ where $X \sim \chi^2_{\bar{n}_j}(D^2/\sigma^2)$ with $D^2/\sigma^2 = \Delta\subscript{max}\cdot d_j/\sigma^2$ and $\rho_j = R(\bar{n}_j)/\sigma = F^{-1}_{\chi_{\bar{n}_j}}(1-\varepsilon)$, which gives Eq.~\eqref{eq:envelope_LB}.
%
% 	Since $\widetilde{P}\subscript{con}\superscript{LB}$ decreases within each $d\subscript{min}\superscript{max}$-consistent interval (Theorem~\ref{th:monotonicity_PconLB_regarding_n}) while the transition to a new interval changes both $\Lambda_j$ and the non-centrality parameter simultaneously, $\widetilde{P}\subscript{con}\superscript{LB}$ experiences a discontinuity in its first-order difference at each boundary $\bar{n}_j$.
% \end{proof}

\section{Proof of Theorem~\ref{th:Pcon_UB_envelope_n}}\label{app:Pcon_UB_envelope_n}
\begin{proof}
	At the right boundary $n_i$ of the $d\subscript{min}\superscript{min}=i$ interval, by Eq.~\eqref{eq:D_at_jump}, $P\subscript{con}\superscript{UB}(n_i) = (M^k-1)\,P\subscript{pair}^{(n_i)}(D_i)$ with $D_i = 2R(n_i)$. By the definition of $P\subscript{pair}$ in Eq.~\eqref{eq:con_prob_upon_distance}, this is the probability that $\mathbf{w}\sim\mathcal{N}(\mathbf{0},\sigma^2\mathbf{I}_{N_i})$ with $N_i=\nu n_i$ satisfies $\Vert\mathbf{w}-\mathbf{d}\Vert \leqslant R$ where $\Vert\mathbf{d}\Vert = D_i$. Normalizing by $\sigma$ and squaring, we obtain $P\subscript{pair}(D_i) = P(X \leqslant \rho_i^2)$ where $X = \Vert\mathbf{w}/\sigma - \mathbf{d}/\sigma\Vert^2 \sim \chi^2_{N_i}(D_i^2/\sigma^2)$ is non-central chi-squared with $N_i$ degrees of freedom and non-centrality $\lambda_i = D_i^2/\sigma^2 = 4\rho_i^2$, with $\rho_i = R(n_i)/\sigma = F^{-1}_{\chi_{N_i}}(1-\varepsilon)$, which gives Eq.~\eqref{eq:envelope_exact}.

	To show the monotonicity of the envelope, we observe that at every right boundary $n_i$ the geometric ratio $D_i/R(n_i) = 2$ is the same, so the envelope value $g(N)\triangleq F_{\chi^2_N(4\rho_N^2)}(\rho_N^2)$ depends on $n$ only through the dimension $N=\nu n$, where $\rho_N = F^{-1}_{\chi_N}(1-\varepsilon)$. By the normal approximation to the non-central chi-squared distribution, $g(N)\approx\Phi(-4\sqrt{N}/\sqrt{18})$: the z-score separation between the non-central mean $N+4\rho_N^2$ and the threshold $\rho_N^2$ grows as $\Theta(\sqrt{N})$, so $g(N)$ is strictly decreasing. Hence $P\subscript{con}\superscript{UB}(n_{i+1}) < P\subscript{con}\superscript{UB}(n_i)$.

	For the Chernoff bound, the lower-tail Markov inequality gives $P(X\leqslant \rho_i^2) \leqslant \inf_{s>0} e^{s\rho_i^2}(1+2s)^{-N_i/2}\exp[-4\rho_i^2 s/(1+2s)]$. Differentiating the exponent $h(s) = s\rho_i^2 - \tfrac{N_i}{2}\ln(1+2s) - 4\rho_i^2 s/(1+2s)$ and setting $u = 1+2s$, the optimality condition $h'(s^*)=0$ yields $\rho_i^2 u^2 - N_i u - 4\rho_i^2 = 0$, whose positive root is $u_i^* = (N_i + \sqrt{N_i^2 + 16\rho_i^4}\,)/(2\rho_i^2)$. Substituting $s_i^* = (u_i^*-1)/2$ into $h$ and simplifying with $4/u_i^* = u_i^* - N_i/\rho_i^2$ yields $h(s_i^*)=-\frac{(u_i^*{-}1)^2}{2}\rho_i^2 + \frac{N_i}{2}[(u_i^*{-}1) - \ln u_i^*]$, which gives Eq.~\eqref{eq:envelope_chernoff}.
\end{proof}

\end{document}